\documentclass[12pt, a4paper, reqno]{amsart}
\usepackage{mathrsfs}
\usepackage{color}
\usepackage{amsmath,amscd,amssymb,latexsym}
\usepackage[all]{xy}

\input xypic

\evensidemargin=\oddsidemargin
\DeclareMathVersion{can}
\DeclareMathAlphabet{\can}{OT1}{cmss}{m}{n}
\newtheorem{theorem}{Theorem}[section]
\newtheorem{lemma}[theorem]{Lemma}
\newtheorem{proposition}[theorem]{Proposition}

\theoremstyle{definition}

\theoremstyle{fact}

\theoremstyle{conjecture}

\numberwithin{equation}{section}

\newcommand{\C}{{\mathcal C}}
\newcommand{\F}{\mathbb{F}}
\newcommand{\f}{\mathbb{F}}

\newcommand{\Tr}{\mathrm{Tr}_q^{q^2}}

\begin{document}
\title[Repeated-root  cyclic codes of length $5p^s$]{On the Hamming distances of repeated-root cyclic
codes of length $5p^s$}
\author[X. Li]{Xia Li}
\address{\rm Department of Mathematics, Nanjing University of Aeronautics and Astronautics,
Nanjing, 211100, P. R. China; State Key Laboratory of Cryptology, P. O. Box 5159, Beijing, 100878, P. R. China}
 \email{lixia4675601@163.com}

\author[Q. Yue]{Qin Yue}
\address{\rm Department of Mathematics, Nanjing University of Aeronautics and Astronautics,
Nanjing, 211100, P. R. China; State Key Laboratory of Cryptology, P. O. Box 5159, Beijing, 100878, P. R. China}
\email{yueqin@nuaa.edu.cn}
\thanks{The paper was supported by National Natural Science Foundation of China (No.
 61772015).}

\begin{abstract}Due to the wide applications in consumer electronics, data storage systems and communication systems, cyclic codes have been an interesting research topic in coding theory. In this paper, let $p$ be a prime with  $p\ge 7$.  We determine the weight distributions of all cyclic codes of length $5$ over $\f_q$ and the Hamming distances of all repeated-root  cyclic codes of length $5p^s$ over $\F_q$, where $q=p^m$  and both  $s$ and $m$ are positive integers.
Furthermore, we find
all MDS cyclic codes of length $5p^s$ and take quantum synchronizable codes from repeated-root
cyclic codes of length $5p^s$.

\end{abstract}


\maketitle

\section{Introduction}

Cyclic codes have important applications in consumer electronics, data storage systems and communication systems as they have efficient encoding and decoding algorithms compared with the linear block codes \cite{T07}. They also have applications in cryptography \cite{C05, D05} and sequence design \cite{D10}.

Repeated-root cyclic codes were first investigated in the 1990s by Castagnoli in \cite {C91} and Van Lint in \cite{V91}. Dinh determined the Hamming distances of all negacyclic
and cyclic codes of length $p^s$ over the finite field $\F_q$ in \cite{D08} and the Hamming distances of all $(\alpha+u\beta)$-constacyclic codes of length $p^s$ over $\F_{p^m}+u\F_{p^m}$ in \cite{D2010}. Kai and Zhu determined the distances of cyclic codes of length $2^e$ over $\Bbb Z_4$ in \cite{K10}. Dinh determined the generator polynomials of all constacyclic codes and their dual codes of length $2p^s$, $3p^s$ and $6p^s$  over $\F_q$ in \cite{D2012, D2013, D2014}. Chen et al. studied repeated-root constacyclic codes of length $l^tp^s$ over $\F_q$ in \cite{B12}. In 2014, Chen  et al. studied all constacyclic codes of length $lp^s$ over $\F_q$ in \cite{C2014}, where $l$ is a prime different from $p$. In 2015, Raka  considered repeated-root constacyclic codes of length $l^tp^s$ over $\F_q$ in \cite{M2015}. Sharma explicitly determined the generator polynomials of all repeated-root constacyclic codes of length $l^tp^s$ over $\F_{p^m}$ and their dual codes in \cite{S2015}. Chen et al. investigated all constacyclic codes of length $2l^mp^s$ over $\F_q$ of characteristic $p$ in \cite{C2015}. Sharma and Rani gave a method to compute repeated-root constacyclic codes of length $np^s$ over $\F_q$, for any positive integer $n$ coprime to $p$, however the generator polynomials are not determined explicitly, in \cite{S2016}. Liu et al. discussed all constacyclic codes of length $3lp^s$ over $\F_q$, where $p\neq3$ is any prime and $l\neq3$ is an odd prime with $\gcd(l,p)=1$ in \cite{L2016}.

Quantum synchronizable codes has proved to be
functioning well, which allows for extracting the information
about the magnitude and direction of misalignment and
simultaneously correcting the Pauli errors on qubits, with
nondisturbing measurement involved. In this coding scheme,
quantum synchronizable codes can be constructed from a pair
of dual-containing cyclic codes with one contained in the other. Luo and Ma \cite{L2018} exploited repeated-root cyclic codes of
lengths $p^s$ and $lp^s$ over $\F_q$ to obtain two new kinds of
quantum synchronizable codes with the highest misalignment
tolerance, where $s>1$, $l>2$ are integers, and $p$ is the odd
characteristic of $\F_q$ such that $\gcd(l, p) = 1$.

In this paper, we study all repeated-root  cyclic codes of length $5p^s$ over the finite field $\F_q$, where $p\geq 7$ is a prime. We obtain the weight distributions of  all cyclic codes of length $5$ over $\f_q$ and the Hamming distances of all repeated-root  cyclic codes of length $5p^s$ over the finite field $\F_q$. Furthermore, we find
all MDS cyclic codes of length $5p^s$ and take quantum synchronizable codes from repeated-root
cyclic codes of length $5p^s$.

\section{Preliminaries}
Throughout this paper, suppose that $q=p^m$, where $p$ is a prime and $m$ is a positive integer. An $[n,\kappa,d]$ linear code over the finite field  $\F_q$ is an
$\kappa$-dimensional subspace of $\F^n_{q}$ with minimum (Hamming)
distance $d$. Let $A_i$ denote  the number of codewords with Hamming weight
$i$ in a code $\C$ of length $n$. The {weight enumerator} of $\C$ is defined by
$$
1 + A_1z + A_2z^2 + \cdots + A_nz^n.
$$
Accordingly, the sequence $(1,A_1, \cdots, A_n)$ is called the weight distribution of $\C$. A linear code $\C$ is said to be a $t$-weight code if the number of nonzero $A_i$ in the sequence $(A_1,A_2,\cdots,A_n)$ is equal to $t$. This code is called cyclic if
$(c_0,c_1, \cdots, c_{n-1}) \in \C$ implies $(c_{n-1}, c_0, c_1, \cdots, c_{n-2})
\in\mathcal{C}$.
Let $\gcd(n,q)=1$, by identifying a vector $(c_0,c_1, \cdots, c_{n-1}) \in \F^n_{q}$
with
$$
c_0+c_1x+c_2x^2+ \cdots + c_{n-1}x^{n-1} \in \F_q[x]/(x^n-1),
$$
a code $\mathcal{C}$ of length $n$ over $\F_q$ corresponds to a subset of $\F_q[x]/(x^n-1)$.
The linear code $\C$ is cyclic if and only if such corresponding subset in $\F_q[x]/(x^n-1)$
is an ideal of the polynomial residue class ring $\F_q[x]/(x^n-1)$.
It is well known that every ideal of $\F_q[x]/(x^n-1)$ is principal. Let $\mathcal{C}=(g(x))$,
where $g(x)$ is monic and has the least
degree. Then $g(x)$ is called the {generator polynomial} and
$h(x)=(x^n-1)/g(x)$ is referred to as the {parity-check} polynomial of
$\mathcal{C}$. Further, the definition of the dual code of code $\C$ is as follows:
 $$\C^\bot=\{x\in\F^n_{q}\mid x\cdot y=0, \forall y\in \C\},$$
where $x\cdot y$ denotes the Euclidean inner product of $x$ and $y$ in $\F^n_{q}$ .

Let $\C=<g(x)>$ be a repeated-root cyclic code of length $5p^s$ over $\f_q$, where $p$ is a prime  with $p\ge 7$ and $s$
is a  positive integer.  Let $\C_{i,5}$ be the
cyclotomic coset of $i$ modulo $5$ over $\F_q$ and denote by $T_5$ the set of representatives of all $q$-ary cyclotomic cosets. Let $\omega$ be
a primitive $5$-th root of unity in $\F_q$  and $M_i(x)=\prod\limits_{j\in\C_{i,5}}(x-\omega^j)$ be the minimal polynomial of $\omega^i$ over $\F_q$. Suppose that $g(x)=\prod\limits_{i\in T_5}M_i(x)^{t_i}$,
 fix a value $t$, $0\leq t\leq p^s-1$, $\C_t$ is defined to be a simple-root cyclic code of length $5$ over $\f_q$ with generator polynomial $g_{t}(x)$ as the product of those irreducible factors $M_{i}(x)$ of $g(x)$ that occur with multiplicity $t_i>t$. If this product is equal to  $x^5-1$, then $\C_t$ contains only the all-zero codeword and we set $d_H(\C_t)=\infty$. If all $t_i$ satify $t_i\leq t$, then, by way of convention, $g_{t}(x)=1$ and $d_H(\C_t)=1$.

\begin{lemma}[\cite{B18}]
Let $\C=<g(x)>$ be a repeated-root cyclic code of length $5p^s$ over $\f_q$, where $p$ be a prime with $p\ge 7$ and $s$
is a  positive integers. Then
\begin{equation}\label{25}
  d_H(\C) = \min\{P_td_H(\C_t) : t \in T\},
\end{equation}
where $T=\{t:0\leq t\leq p^s-1\}$, $t=\sum^{s-1}\limits_{m=0}t_{m}p^{m}$, $0\leq t_m\leq p-1$, $P_t=\prod_{m=0}^{s-1}(t_m+1)=wt_{H}((x-1)^t)$.
\end{lemma}
\begin{lemma}[\cite{L2013}]

  Let $\beta, \tau$ be integers such that $0 \leq \beta \leq p -2$ and $0 \leq \tau \leq s -1$, then

\begin{equation}\label{24}
   \min\{P_t : t \geq l \mbox{ and }  t \in T \} = (\beta+2)p^\tau,
\end{equation}
where $l$ is an integer such that $p^s-p^{s-\tau}+\beta p^{s-\tau-1}+1\leq l\leq p^s-p^{s-\tau}+(\beta+1) p^{s-\tau-1}$.
\end{lemma}

In the following, we consider the irreducible factorization of $x^5-1$ over $\Bbb F_q$ in three cases.

Suppose that  $5\nmid(q^2-1)$.  Then $x^5-1=(x-1)\phi_{5}(x)$, where $\phi_{5}(x)=x^4+x^3+x^2+x+1$ is irreducible over $\Bbb F_q$.

Suppose that  $5\mid(q+1)$ and $5\nmid(q-1)$.  Then $x^5-1=(x-1)f_{1}(x)f_{2}(x)$, where $f_{1}(x)=(x-\omega)(x-\omega^4)$ and $ f_{2}(x)= (x-\omega^2)(x-\omega^3)$ are irreducible over $\Bbb F_q$ and $\omega\in \Bbb F_{q^2} $ is a $5$-th root of unity.

Suppose that  $5\mid(q-1)$. Then $x^5-1=(x-1)(x-\omega)(x-\omega^2)(x-\omega^3)(x-\omega^4)$, where $\omega\in \Bbb F_q$ is a $5$-th root of unity.

 \begin{lemma}[\cite{V2003}]\label{l1}
 Let $\C$ ba an $[n, \kappa]$ code over $\F_q$ with weight enumerator $A(z)$ and let $B(z)$ be the weight
enumerator of $\C^\bot$. Then
$$B(z)=q^{-\kappa}(1+(q-1)z)^{n}A(\frac{1-z}{1+(q-1)z}).$$
 \end{lemma}
\begin{theorem}
Let $\C_1$ be  a  cyclic code of length $5$ over $\f_q$.
\begin{itemize}
\item Case 1. $5\nmid(q^2-1)$.

If $\C_1=<x-1>$, then $\C_1$ is a $[5, 4]$ cyclic code over $\f_q$. Furthermore, the weight distribution of $\C_1$ is given by Table \ref{tab-CG11}.

If $\C_1=<\phi_{5}(x)>$, then $\C_1$ is a $[5, 1]$ cyclic code over $\f_q$. Furthermore, the weight distribution of $\C_1$ is given by Table \ref{tab-CG14}.

\item Case 2. $5\mid(q+1)$ and $5\nmid(q-1)$.

If $\C_1=<x-1>$, then $\C_1$ is a  $[5, 4]$ cyclic code over $\f_q$. Furthermore, the weight distribution of $\C_1$ is given by Table \ref{tab-CG11}.

If $\C_1=<f_1(x)>$ or $<f_2(x)>$, then $\C_1$ is a $[5, 3]$ cyclic code over $\f_q$. Furthermore, the weight distribution of $\C_1$ is given by Table \ref{tab-CG12}.

If $\C_1=<(x-1)f_1(x)>$ or $<(x-1)f_2(x)>$, then $\C_1$ is a $[5, 2]$ cyclic code over $\f_q$. Furthermore, the weight distribution of $\C_1$ is given by Table \ref{tab-CG13}.

If $\C_1=<f_1(x)f_2(x)>$, then $\C_1$ is a $[5, 1]$ cyclic code over $\f_q$. Furthermore, the weight distribution of $\C_1$ is given by Table \ref{tab-CG14}.

\item Case 3. $5\mid(q-1)$.

 If $\C_1=<x-1>$, then $\C_1$ is a $[5, 4]$ cyclic code over $\f_q$. Furthermore, the weight distribution of $\C_1$ is given by Table \ref{tab-CG11}.

 If  $\C_1=<(x-w^{u_1})(x-w^{u_2})>(0\leq u_1< u_2\leq 4)$, then $\C_1$ is a $[5, 3]$ cyclic code over $\f_q$. Furthermore, the weight distribution of $\C_1$ is given by Table \ref{tab-CG12}.

 If $\C_1=<(x-w^{u_1})(x-w^{u_2})(x-w^{u_3})>(0\leq u_1< u_2< u_3\leq 4)$, then $\C_1$ is a $[5, 2]$ cyclic code over $\f_q$. Furthermore, the weight distribution of $\C_1$ is given by Table \ref{tab-CG13}.

If $\C_1=<(x-w^{u_1})(x-w^{u_2})(x-w^{u_3})(x-w^{u_4})>(0\leq u_1<u_2< u_3< u_4\leq 4)$, then $\C_1$ is a $[5, 1]$ cyclic code over $\f_q$. Furthermore, the weight distribution of $\C_1$ is given by Table \ref{tab-CG14}.

 \end{itemize}
 \end{theorem}
 \begin{proof}
 In Case 1, if $\C_1=<\phi_{5}(x)>$, then $\C^\bot_1=<x-1>$. It is clear that the $\C_1$ is a $[5,1]$ cyclic code and $\C^\bot_1$ is a $[5,4]$ cyclic code. And $A(z)=1+(q-1)z^5$.  By Lemma \ref{l1}, we can get $B(z)=1+10(q-1)z^2+10(q-1)(q-2)z^3+5(q-1)(q^2-3q+3)z^4+(q-1)(q^3-4q^2+6q-4)z^5$.

 In Case 2, if $\C_1=<(x-1)f_{2}(x)>$,  then $\C^\bot_1=<f_1(x)>$.  It is clear that  $\C_1$ is a $[5,2]$ cyclic code and $\C^\bot_1$ is a $[5,3]$ cyclic code. By Delsarte's Theorem \cite{D1975}, we can get
$$\C_t=\{\mathbf{c}(a):a\in \f_{q^2}\}$$
where the codeword
$$\mathbf{c}(a)={(\Tr(a\omega^{-i}))}^4_{i=0}.$$

Let $\mathbf{c}(a)= (a_0, a_1, a_2, a_3, a_4)$ be a codeword
in $\C_t$, where $a\in \f_{q^2}$.
Then the  weight of $\mathbf{c}(a)$ is equal to $5-N(a)$, where
 $N(a)=\sharp\{0\leq i\leq 4:\Tr(a\omega^{-i})=0\}$.

First,  suppose that $a=0$. It is clear that $wt_{H}(\mathbf{c}(a))=0$.

Next, suppose that $a\neq0$.

When $i=0$, $\Tr(a)=0\Longleftrightarrow a+a^q=0\Longleftrightarrow a^{q-1}=-1$;

When $i=1$, $\Tr(a\omega^{-1})=0\Longleftrightarrow a\omega^{-1}+a^q\omega=0\Longleftrightarrow a^{q-1}=-\omega^{-2}$;

When $i=2$, $\Tr(a\omega^{-2})=0\Longleftrightarrow a\omega^{-2}+a^q\omega^{2}=0\Longleftrightarrow a^{q-1}=-\omega^{-4}$;

When $i=3$, $\Tr(a\omega^{-3})=0\Longleftrightarrow a\omega^{-3}+a^q\omega^{3}=0\Longleftrightarrow a^{q-1}=-\omega^{-1}$;

When $i=4$, $\Tr(a\omega^{-4})=0\Longleftrightarrow a\omega^{-4}+a^q\omega^{4}=0\Longleftrightarrow a^{q-1}=-\omega^{-3}$.

It is clear that $N(a)=0$ or $1$. Hence, $A(z)=1+5(q-1)z^4+(q-1)(q-4)z^5$.  By Lemma \ref{l1}, we can get $B(z)=1+10(q-1)z^3+5(q-1)(q-3)z^4+(q-1)(q^2-4q+6)z^5$.

By the similar arguments as above, we can get $A(z)=1+5(q-1)z^4+(q-1)(q-4)z^5$ when $\C_1=<(x-1)f_{1}(x)>$ and $A(z)=1+10(q-1)z^3+5(q-1)(q-3)z^4+(q-1)(q^2-4q+6)z^5$ when $\C_1=<f_{2}(x)>$.

If $\C_1=<f_1(x)f_{2}(x)>$,  then $\C^\bot_1=<x-1>$. The proof is similar to that of Case 1. So we omit it.

In Case 3, the proof is similar to that of Case 2. So we omit it.
 \end{proof}
\begin{table}
\begin{center}
\caption{The weight distribution of the code $\C_1=<(x-1)>$ }\label{tab-CG11}
\begin{tabular}{c|c}
Weight  &  Multiplicity   \\ \hline
 $ 0$ & $1$\\ \hline
 $2$          &  $10(q-1)$ \\\hline
$3$ & $10(q-1)(q-2)$\\
$4$ & $5(q-1)(q^2-3q+3)$\\
$5$ & $(q-1)(q^3-4q^2+6q-4)$\\
\end{tabular}
\end{center}
\end{table}
\begin{table}
\begin{center}
\caption{The weight distribution of the code $\C_1=<f_1(x)>$ or $<f_2(x)>$, $<(x-w^{u_1})(x-w^{u_2})>(0\leq u_1< u_2\leq 4)$ }\label{tab-CG12}
\begin{tabular}{c|c}
Weight  &  Multiplicity   \\ \hline
 $ 0$ & $1$\\ \hline
 $3$ & $10(q-1)$\\
$4$ & $5(q-1)(q-3)$\\
$5$ & $(q-1)(q^2-4q+6)$\\
\end{tabular}
\end{center}
\end{table}

\begin{table}
\begin{center}
\caption{The weight distribution of the code $\C_1=<(x-1)f_{2}(x)>$ or $<(x-1)f_{1}(x)>$, $<(x-w^{u_1})(x-w^{u_2})(x-w^{u_3})>(0\leq u_1< u_2< u_3\leq 4)$ }\label{tab-CG13}
\begin{tabular}{c|c}
Weight  &  Multiplicity   \\ \hline
 $ 0$ & $1$\\ \hline
 $4$          &  $5(q-1)$ \\\hline
$5$ & $(q-1)(q-4)$\\
\end{tabular}
\end{center}
\end{table}

\begin{table}
\begin{center}
\caption{The weight distribution of the code $\C_1=<\phi_{5}(x)>$ or $<f_{1}(x)f_{2}(x)>$, $<(x-w^{u_1})(x-w^{u_2})(x-w^{u_3})(x-w^{u_4})>(0\leq u_1< u_2< u_3< u_4\leq 4)$ }\label{tab-CG14}
\begin{tabular}{c|c}
Weight  &  Multiplicity   \\ \hline
 $ 0$ & $1$\\ \hline
$5$ & $(q-1)$\\
\end{tabular}
\end{center}
\end{table}
\begin{proposition}\label{p1}
Let $t$ be an integer such that $0\leq t\leq p^s-1$. We denote the simple-root cyclic code
$\C_t =< g_t(x) >\subset \f_q[x]/(x^5-1)$ depending on $\C$ and $t$.
\begin{itemize}
\item Case 1. Suppose that  $5\nmid(q^2-1)$. Then
 \begin{equation*}\label{222}
 d_{H}(\C_t)=\left\{
  \begin{array}{ll}

    1,&\mbox{ if }  g_t=1,\\
    2,&\mbox{ if } g_t=x-1,\\
 4,&\mbox{ if } g_t=\phi_{5}(x),\\
\infty,&\mbox{ if } g_t=x^5-1.
\end{array}
\right.
\end{equation*}
\item Case 2. Suppose that  $5\mid(q+1)$ and $5\nmid(q-1)$. Then
\begin{equation*}\label{222}
 d_{H}(\C_t)=\left\{
  \begin{array}{ll}

    1,&\mbox{ if }  g_t=1,\\
    2,&\mbox{ if } g_t=x-1,\\
 3,&\mbox{ if } g_t=f_{1}(x) \mbox{ or } f_{2}(x),\\
4,&\mbox{ if } g_t=(x-1)f_{1}(x)\mbox{ or }(x-1)f_{2}(x),\\
5,&\mbox{ if } g_t=f_{1}(x)f_{2}(x),\\
\infty, &\mbox{ if } g_t=x^5-1.
\end{array}
\right.
\end{equation*}
\item Case 3. Suppose that  $5\mid(q-1)$.  Then
\begin{equation*}\label{222}
 d_{H}(\C_t)=\left\{
  \begin{array}{ll}

    1,&\mbox{ if }  g_t=1,\\
    2,&\mbox{ if } g_t=x-1,\\
 3,&\mbox{ if } g_t=(x-w^{u_1})(x-w^{u_2})(0\leq u_1<u_2\leq 4),\\
4,&\mbox{ if } g_t=(x-w^{u_1})(x-w^{u_2})(x-w^{u_3})(0\leq u_1< u_2< u_3\leq 4),\\
5,&\mbox{ if } g_t=(x-w^{u_1})(x-w^{u_2})(x-w^{u_3})(x-w^{u_4})(0\leq u_1< u_2< u_3<u_4\leq 4),\\
\infty, &\mbox{ if } g_t=x^5-1.
\end{array}
\right.
\end{equation*}
\end{itemize}
\end{proposition}

\section{The minimum hamming distances of all repeated-root  cyclic codes of length $5p^s$}

It is clear that all cyclic codes of length $5p^s$ have  form $\C=<(x-1)^i{\phi_{5}(x)}^j>$ with $0 \leq i, j\leq p^s$, when $5\nmid(q^2-1)$. If $i=0$ and $j=0$, then $\C=\f_{q}[x]/(x^{5p^s}-1)$. If $i=p^s$ and $j=p^s$, then $\C=\{\mathbf{0}\}$. These are the trivial cyclic codes of length $5p^s$ over $\f_q$. Next, we determine the minimum Hamming distances of all non-trivial cyclic codes of length $5p^s$ over $\f_q$  when $i\geq j$. For $j>i$, the results of the minimum Hamming distances can be computed by the similar technique.

\begin{lemma}
Let $i, j$ be integers such that $0< i\leq p^s$ and $j=0$,  then $d_H(\C)=2$.
\end{lemma}
\begin{proof}
 According to Proposition \ref{p1} and $(\ref{25})$, we can get $ d_{H}(\C)=2$.
\end{proof}
\begin{lemma}
Let $i, j$ be integers with $p^s-p^{s-\tau_{0}}+\beta_{0} p^{s-\tau_{0}-1}+1\leq i\leq p^s-p^{s-\tau_{0}}+(\beta_{0}+1) p^{s-\tau_{0}-1}$ and $p^s-p^{s-\tau_{1}}+\beta_{1}p^{s-\tau_{1}-1}+1\leq j\leq p^s-p^{s-\tau_{1}}+(\beta_{1}+1) p^{s-\tau_{1}-1}$, where $0 \leq \beta_{0}, \beta_{1}\leq p -2$ and $0 \leq\tau_{1} \leq\tau_{0}\leq s -1$.  Let $\C=<(x-1)^i{\phi_{5}(x)}^j>$, then $d_H(\C) =\min\{(\beta_{0}+2)p^{\tau_{0}}, 2(\beta_{1}+2)p^{\tau_{1}} \}$.
\end{lemma}
\begin{proof}

If $t \geq i \geq j$, then  $\min\{P_t :  t \in T  \mbox{ and } t\geq i\geq j\} = (\beta_{0}+2)p^{\tau_{0}}$ by $(\ref{24})$, and $d_H(\C_t)=1$ by Proposition $\ref{p1}$. Thus,
\begin{equation*}
 \min\{P_td_H(\C_t) : t \in T, t\geq i\geq j\}=(\beta_{0}+2)p^{\tau_{0}}.
\end{equation*}

If $ i> t\geq j$, then  $\min\{P_t :  t \in T  \mbox{ and } i> t\geq j\} \geq (\beta_{1}+2)p^{\tau_{1}}$ by $(\ref{24})$, and $d_H(\C_t)=2$ by Proposition $\ref{p1}$. Thus,
\begin{equation*}
 \min\{P_td_H(\C_t) : t \in T,i>t\geq j\}\geq2(\beta_{1}+2)p^{\tau_{1}}.
\end{equation*}

In summary, $d_H(\C)=\min\{(\beta_{0}+2)p^{\tau_{0}}, 2(\beta_{1}+2)p^{\tau_{1}}\}$. The proof of this lemma is done.

\end{proof}

By the similar arguments as above, we can obtain the following lemmas immediately.

\begin{lemma}
Let $j$ be integer with $p^s-p^{s-\tau_{1}}+\beta_{1}p^{s-\tau_{1}-1}+1\leq j\leq p^s-p^{s-\tau_{1}}+(\beta_{1}+1) p^{s-\tau_{1}-1}$, where $0 \leq \beta_{1}\leq p -2$ and $0 \leq\tau_{1}\leq s -1$.  Let $\C=<(x-1)^{p^s}{\phi_{5}(x)}^j>$, then
$d_H(\C) =2(\beta_{1}+2)p^{\tau_{1}}$.
\end{lemma}

According to the lemmas as above, we have the following theorem.
\begin{theorem}
Let $p\geq 7$ be a prime satisfying $5\nmid(q^2-1)$, $0 \leq \beta_{0},\beta_{1}\leq p -2$ and $0 \leq\tau_{1} \leq\tau_{0}\leq s -1$. Then cyclic codes of length $5p^s$ have  form $\C=<(x-1)^i{\phi_{5}(x)}^j>$. If  $0 \leq j\leq i\leq p^s$, then the minimum Hamming distances of $\C$ are given in Table \ref{Table:31}. If  $0 \leq i \leq j \leq p^s$, then the minimum Hamming distances of $\C$ are also shown in Table \ref{Table:32}.\end{theorem}
\begin{table}\tiny
\begin{center}
\caption{The minimum Hamming distances of $\C=<(x-1)^i{\phi_{5}(x)}^j>$, where $0 \leq j \leq i \leq p^s$.}\label{Table:31}
\begin{tabular}{c|c|c}
    \hline  $i$ & $j$  & $d_H(\C)$ \\\hline
       $0$   &$0$ &$1$\\\hline
        $0<i\leq p^s$ & $j=0$  &$2$\\\hline
      $\begin{array}{c}
  p^s-p^{s-\tau_{0}}+\beta_{0} p^{s-\tau_{0}-1}+1\leq i\\
  \leq p^s-p^{s-\tau_{0}}+(\beta_{0}+1) p^{s-\tau_{0}-1}
\end{array}$
 &  $\begin{array}{c}
  p^s-p^{s-\tau_{1}}+\beta_{1} p^{s-\tau_{1}-1}+1\leq j\\
  \leq p^s-p^{s-\tau_{1}}+(\beta_{1}+1) p^{s-\tau_{1}-1}
\end{array}$ & $\min\{(\beta_{0}+2)p^{\tau_{0}}, 2(\beta_{1}+2)p^{\tau_{1}}\}$
  \\ \hline
$i=p^s$&$\begin{array}{c}
  p^s-p^{s-\tau_{1}}+\beta_{1} p^{s-\tau_{1}-1}+1\leq j\\
  \leq p^s-p^{s-\tau_{1}}+(\beta_{1}+1) p^{s-\tau_{1}-1}
\end{array}$ &$2(\beta_{1}+2)p^{\tau_{1}}$\\ \hline
$i=p^s$&$j=p^s$&$0$\\ \hline
\end{tabular}
\end{center}
\end{table}
\begin{table}\tiny
\begin{center}
\caption{The minimum Hamming distances of $\C=<(x-1)^i{\phi_{5}(x)}^j>$, where $0 \leq i \leq j\leq p^s$.}\label{Table:32}
\begin{tabular}{c|c|c}
    \hline  $i$ & $j$  & $d_H(\C)$ \\\hline
       $0$   &$0$ &$1$\\\hline
        $i=0$ & $0<j\leq p^{s-1}$  &$2$\\\hline
 $i=0$ & $p^{s-1}<j\leq2 p^{s-1}$  &$3$\\\hline
 $i=0$ & $2p^{s-1}<j\leq p^s$  &$4$\\\hline
      $\begin{array}{c}
  p^s-p^{s-\tau_{1}}+\beta_{1} p^{s-\tau_{1}-1}+1\leq i\\
  \leq p^s-p^{s-\tau_{1}}+(\beta_{1}+1) p^{s-\tau_{1}-1}
\end{array}$
 &  $\begin{array}{c}
  p^s-p^{s-\tau_{0}}+\beta_{0} p^{s-\tau_{0}-1}+1\leq j\\
  \leq p^s-p^{s-\tau_{0}}+(\beta_{0}+1) p^{s-\tau_{0}-1}
\end{array}$ & $\min\{(\beta_{0}+2)p^{\tau_{0}}, 4(\beta_{1}+2)p^{\tau_{1}}\}$
  \\ \hline
$\begin{array}{c}
  p^s-p^{s-\tau_{1}}+\beta_{1} p^{s-\tau_{1}-1}+1\leq i\\
  \leq p^s-p^{s-\tau_{1}}+(\beta_{1}+1) p^{s-\tau_{1}-1}
\end{array}$&$j=p^s$ &$4(\beta_{1}+2)p^{\tau_{1}}$\\ \hline
$i=p^s$&$j=p^s$&$0$\\ \hline
\end{tabular}
\end{center}
\end{table}

 It is clear that all cyclic codes of length $5p^s$ have  form $\C=<(x-1)^i{f_{1}(x)}^j{f_{2}(x)}^k>$ with $0 \leq i, j, k\leq p^s$, when $5\mid(q+1)$ and $5\nmid(q-1)$. If $i=j=k=0$, then $\C=\f_{q}[x]/(x^{5p^s}-1)$, $d_H(\C)=1$. If $i=j=k=p^s$, then $\C=\{\mathbf{0}\}$, $d_H(\C)=0$. These are the trivial cyclic codes of length $5p^s$ over $\f_q$. Next, we determine the minimum Hamming distances of all non-trivial cyclic codes of length $5p^s$ over $\f_q$
Next, we will begin to discuss the minimum Hamming distances of all cyclic codes when $i\geq j\geq k$. For $k\geq j\geq i$, the results of the minimum Hamming distances can be computed by the similar technique.

We here detemine the Hamming distances of $\C=<(x-1)^i{f_{1}(x)}^j{f_{2}(x)}^k>$ for $0 \leq k\leq j \leq i \leq p^s$. Using the similar way, we show the Hamming distances of $\C=<(x-1)^i{f_{2}(x)}^j{f_{1}(x)}^k>$ for $0 \leq k\leq j \leq i \leq p^s$.

\begin{lemma}\label{27}
Let $i, j$ be integers such that $0< i\leq p^s$ and $0\leq j\leq p^s$, $k=0$, then
\begin{equation*}\label{222}
 d_{H}(\C)=\left\{
  \begin{array}{ll}

    2,&\mbox{ if }  0< i\leq p^{s-1}, 0\leq j\leq p^{s-1}, \mbox{ or } p^{s-1}<i\leq p^s,  j=0,\\
    3,&\mbox{ if }  p^{s-1}<i\leq 2p^{s-1}, 0<j\leq p^s,\\
 4,&\mbox{ if } 2p^{s-1}<i\leq p^s,  0<j\leq p^s.
\end{array}
\right.
\end{equation*}
\end{lemma}
\begin{proof}
\begin{itemize}
\item Case 1. $0< i\leq p^{s-1}$.

 According to Proposition \ref{p1} and $(\ref{25})$, we can get $ d_{H}(\C)=2$;

\item Case 2. $ p^{s-1} <i\leq p^s$.

If $ p^{s-1}<i\leq p^{s}$ and $j=0$, then according to Proposition \ref{p1} and $(\ref{25})$, we can get $ d_{H}(\C)=2$;

If $ p^{s-1}<i\leq 2p^{s-1}$ and $0<j\leq p^s$,  then according to Proposition \ref{p1} and $(\ref{25})$, we can get $ d_{H}(\C)=3$;

If $ 2p^{s-1}<i\leq p^{s}$ and $0<j\leq p^s$, then according to Proposition \ref{p1} and  $(\ref{25})$, we can get $ d_{H}(\C)=4$;

\end{itemize}
According to the cases as above, the proof is done.
\end{proof}
\begin{lemma}\label{28}
Let $i, j, k$ be integers with $p^s-p^{s-\tau_{0}}+\beta_{0} p^{s-\tau_{0}-1}+1\leq i\leq p^s-p^{s-\tau_{0}}+(\beta_{0}+1) p^{s-\tau_{0}-1}$, $p^s-p^{s-\tau_{1}}+\beta_{1}p^{s-\tau_{1}-1}+1\leq j\leq p^s-p^{s-\tau_{1}}+(\beta_{1}+1) p^{s-\tau_{1}-1}$ and $p^s-p^{s-\tau_{2}}+\beta_{2} p^{s-\tau_{2}-1}+1\leq k\leq p^s-p^{s-\tau_{2}}+(\beta_{2}+1) p^{s-\tau_{2}-1}$  where $0 \leq \beta_{0}, \beta_{1}, \beta_{2}\leq p -2$ and $0\leq\tau_{2} \leq\tau_{1} \leq\tau_{0}\leq s -1$.  Let $\C=<(x-1)^i{f_{1}(x)}^j{f_{2}(x)}^k>$, then $d_H(\C) =\min\{ (\beta_{0}+2)p^{\tau_{0}}, 2(\beta_{1}+2)p^{\tau_{1}}, 4(\beta_{2}+2)p^{\tau_{2}}\}$.
\end{lemma}
\begin{proof}

If $t\geq i\geq j\geq k$, then   $\min\{P_t :  t \in T  \mbox{ and } t\geq i\geq j\geq k\} = (\beta_{0}+2)p^{\tau_{0}}$ by $(\ref{24})$, and $d_H(\C_t)=1$ by Proposition $\ref{p1}$. Thus,
\begin{equation*}
 \min\{P_td_H(\C_t) : t \in T, t\geq i\geq j\geq k\}=(\beta_{0}+2)p^{\tau_{0}}.
\end{equation*}

If $i>t\geq j\geq k$, then   $\min\{P_t :  t \in T  \mbox{ and } i>t\geq j\geq k\}\geq(\beta_{1}+2)p^{\tau_{1}}$ by $(\ref{24})$, and $d_H(\C_t)=2$ by Proposition $\ref{p1}$. Thus,
\begin{equation*}
 \min\{P_td_H(\C_t) : t \in T, i>t\geq j\geq k\}\geq2(\beta_{1}+2)p^{\tau_{1}}.
\end{equation*}

If $i\geq j>t\geq k$, then   $\min\{P_t :  t \in T  \mbox{ and } i\geq j>t\geq k\}\geq(\beta_{2}+2)p^{\tau_{2}}$ by
 $(\ref{24})$, and $d_H(\C_t)=4$ by Proposition $\ref{p1}$. Thus,
\begin{equation*}
 \min\{P_td_H(\C_t) : t \in T, i>t\geq j\geq k\}\geq4(\beta_{2}+2)p^{\tau_{2}}.
\end{equation*}

In summary, $d_H(\C) =\min\{ (\beta_{0}+2)p^{\tau_{0}}, 2(\beta_{1}+2)p^{\tau_{1}}, 4(\beta_{2}+2)p^{\tau_{2}}\}$. This completes the proof of this lemma.
\end{proof}

By the similar arguments as above, we can obtain the following lemmas immediately.

\begin{lemma}
Let $i, j, k$ be integers with $i=p^s$, $ p^s-p^{s-\tau_{1}}+\beta_{1}p^{s-\tau_{1}-1}+1\leq j\leq p^s-p^{s-\tau_{1}}+(\beta_{1}+1) p^{s-\tau_{1}-1}$ and $p^s-p^{s-\tau_{2}}+\beta_{2}p^{s-\tau_{2}-1}+1\leq k\leq p^s-p^{s-\tau_{2}}+(\beta_{2}+1) p^{s-\tau_{2}-1}$  where $0 \leq \beta_{1}\leq\beta_{2}\leq p -2$ and $0 \leq\tau_{2}\leq\tau_{1}\leq s -1$.  Let  $\C=<(x-1)^{p^s}{f_{1}(x)}^j{f_{2}(x)}^k>$, then
$d_H(\C) =\min\{2(\beta_{1}+2)p^{\tau_{1}}, 4(\beta_{2}+2)p^{\tau_{2}}\}$.
\end{lemma}
\begin{lemma}
Let $i, j, k$ be integers with $i=p^s, j=p^s$ and $p^s-p^{s-\tau_{2}}+\beta_{2}p^{s-\tau_{2}-1}+1\leq k\leq p^s-p^{s-\tau_{2}}+(\beta_{2}+1) p^{s-\tau_{2}-1}$  where $0 \leq \beta_{2}\leq p -2$ and $0 \leq\tau_{2}\leq s -1$.  Let  $\C=<(x-1)^{p^s}{f_{1}(x)}^{p^s}{f_{2}(x)}^k>$, then
$d_H(\C) =4(\beta_{2}+2)p^{\tau_{2}}$.
\end{lemma}

We here detemine the Hamming distances of $\C=<{f_{1}(x)}^i(x-1)^j{f_{2}(x)}^k>$ for $0 \leq k\leq j \leq i \leq p^s$.

\begin{lemma}\label{27}
Let $i, j$ be integers such that $0< i\leq p^s$,  $0\leq j\leq p^s$ and $k=0$, then
\begin{equation*}\label{222}
 d_{H}(\C)=\left\{
  \begin{array}{ll}

    2,&\mbox{ if }  0< i\leq p^{s-1}, 0\leq j\leq p^{s-1}, \\
    3,&\mbox{ if }  p^{s-1}<i\leq 2p^{s-1}, 0<j\leq p^s,\mbox{ or } p^{s-1}<i\leq p^s,  j=0,\\
 4,&\mbox{ if } 2p^{s-1}<i\leq p^s,  0<j\leq p^s.
\end{array}
\right.
\end{equation*}
\end{lemma}
\begin{proof}
\begin{itemize}
\item Case 1. $0< i\leq p^{s-1}$.

 According to Proposition \ref{p1} and $(\ref{25})$, we can get $ d_{H}(\C)=2$;

\item Case 2. $ p^{s-1} <i\leq p^s$.

If $ p^{s-1}<i\leq p^{s}$ and $j=0$, then according to Proposition \ref{p1} and $(\ref{25})$, we can get $ d_{H}(\C)=3$;

If $ p^{s-1}<i\leq 2p^{s-1}$ and $0<j\leq p^s$,  then according to Proposition \ref{p1} and $(\ref{25})$, we can get $ d_{H}(\C)=3$;

If $ 2p^{s-1}<i\leq p^{s}$ and $0<j\leq p^s$, then according to Proposition \ref{p1} and  $(\ref{25})$, we can get $ d_{H}(\C)=4$;

\end{itemize}
According to the cases as above, the proof is done.
\end{proof}

By the similar arguments as  Lemma \ref{28}, we can obtain the following lemmas immediately.
\begin{lemma}
Let $i, j, k$ be integers with $p^s-p^{s-\tau_{0}}+\beta_{0} p^{s-\tau_{0}-1}+1\leq i\leq p^s-p^{s-\tau_{0}}+(\beta_{0}+1) p^{s-\tau_{0}-1}$, $p^s-p^{s-\tau_{1}}+\beta_{1}p^{s-\tau_{1}-1}+1\leq j\leq p^s-p^{s-\tau_{1}}+(\beta_{1}+1) p^{s-\tau_{1}-1}$ and $p^s-p^{s-\tau_{2}}+\beta_{2} p^{s-\tau_{2}-1}+1\leq k\leq p^s-p^{s-\tau_{2}}+(\beta_{2}+1) p^{s-\tau_{2}-1}$  where $0 \leq \beta_{0}, \beta_{1}, \beta_{2}\leq p -2$ and $0\leq\tau_{2} \leq\tau_{1} \leq\tau_{0}\leq s -1$.  Let $\C=<{f_{1}(x)}^i(x-1)^j{f_{2}(x)}^k>$, then $d_H(\C) =\min\{ (\beta_{0}+2)p^{\tau_{0}}, 3(\beta_{1}+2)p^{\tau_{1}}, 4(\beta_{2}+2)p^{\tau_{2}}\}$.
\end{lemma}
\begin{lemma}
Let $i, j, k$ be integers with $i=p^s$, $ p^s-p^{s-\tau_{1}}+\beta_{1}p^{s-\tau_{1}-1}+1\leq j\leq p^s-p^{s-\tau_{1}}+(\beta_{1}+1) p^{s-\tau_{1}-1}$ and $p^s-p^{s-\tau_{2}}+\beta_{2}p^{s-\tau_{2}-1}+1\leq k\leq p^s-p^{s-\tau_{2}}+(\beta_{2}+1) p^{s-\tau_{2}-1}$  where $0 \leq \beta_{1}\leq\beta_{2}\leq p -2$ and $0 \leq\tau_{2}\leq\tau_{1}\leq s -1$.  Let $\C=<{f_{1}(x)}^{p^s}(x-1)^j{f_{2}(x)}^k>$, then
$d_H(\C) =\min\{3(\beta_{1}+2)p^{\tau_{1}}, 4(\beta_{2}+2)p^{\tau_{2}}\}$.
\end{lemma}
\begin{lemma}
Let $i, j, k$ be integers with $i=p^s, j=p^s$ and $p^s-p^{s-\tau_{2}}+\beta_{2}p^{s-\tau_{2}-1}+1\leq k\leq p^s-p^{s-\tau_{2}}+(\beta_{2}+1) p^{s-\tau_{2}-1}$  where $0 \leq \beta_{2}\leq p -2$ and $0 \leq\tau_{2}\leq s -1$.  Let  $\C=<{f_{1}(x)}^{p^s}(x-1)^{p^s}{f_{2}(x)}^k>$, then
$d_H(\C) =4(\beta_{2}+2)p^{\tau_{2}}$.
\end{lemma}
According to the lemmas as above, we have the following theorem.
\begin{theorem}
Let $p\geq7$ be a prime satisfying $5\mid(q+1)$ and $5\nmid(q-1)$, $0 \leq \beta_{0}, \beta_{1}, \beta_{2}\leq p -2$ and $0\leq\tau_{2} \leq\tau_{1} \leq\tau_{0}\leq s -1$. Then cyclic codes of length $5p^s$ have  form $\C=<(x-1)^i{f_{1}(x)}^j{f_{2}(x)}^k>$ with $0 \leq i, j, k\leq p^s$, and  the minimum Hamming distances of $\C$ are given in Table \ref{Table:33}, \ref{Table:331}, \ref{Table:34} and \ref{Table:341}.  \end{theorem}
\begin{table}\tiny
\begin{center}
\caption{The minimum Hamming distances of $\C=<(x-1)^i{f_{1}(x)}^j{f_{2}(x)}^k>$ or $<(x-1)^i{f_{2}(x)}^j{f_{1}(x)}^k>$, where $0 \leq k\leq j \leq i \leq p^s$.}\label{Table:33}
\begin{tabular}{c|c|c|c}
    \hline  $i$ & $j$ & $k$ & $d_H(\C)$ \\ \hline
       $0$   &$0$ &$0$ &$1$\\\hline
        $0<i\leq p^{s-1}$ & $0\leq j\leq p^{s-1}$ &$k=0$  &$2$\\ \hline
$ p^{s-1}<i\leq p^{s}$ & $j=0$ &$k=0$  &$2$\\ \hline
$p^{s-1}<i\leq 2p^{s-1}$ & $0<j\leq p^{s}$ &$k=0$  &$3$\\ \hline
$2p^{s-1}<i\leq p^{s}$ & $0< j\leq p^{s}$ &$k=0$  &$4$\\ \hline
      $\begin{array}{c}
  p^s-p^{s-\tau_{0}}+\\ \beta_{0} p^{s-\tau_{0}-1}+1\leq i\\
  \leq p^s-p^{s-\tau_{0}}+\\(\beta_{0}+1) p^{s-\tau_{0}-1}
\end{array}$
 &  $\begin{array}{c}
  p^s-p^{s-\tau_{1}}+\\ \beta_{1} p^{s-\tau_{1}-1}+1\leq j\\
  \leq p^s-p^{s-\tau_{1}}+\\(\beta_{1}+1) p^{s-\tau_{1}-1}
\end{array}$ &
 $\begin{array}{c}
  p^s-p^{s-\tau_{2}}+\\ \beta_{2} p^{s-\tau_{2}-1}+1\leq k\\
  \leq p^s-p^{s-\tau_{2}}+\\(\beta_{2}+1) p^{s-\tau_{2}-1}
\end{array}$ & $\begin{array}{c}\min\{(\beta_{0}+2)p^{\tau_{0}}, 2(\beta_{1}+2)p^{\tau_{1}},\\ 4(\beta_{2}+2)p^{\tau_{2}}\}\end{array}$
  \\ \hline
$i=p^s$&$\begin{array}{c}
  p^s-p^{s-\tau_{1}}+\\ \beta_{1} p^{s-\tau_{1}-1}+1\leq j\\
  \leq p^s-p^{s-\tau_{1}}+\\(\beta_{1}+1) p^{s-\tau_{1}-1}
\end{array}$ & $\begin{array}{c}
  p^s-p^{s-\tau_{2}}+\\ \beta_{2} p^{s-\tau_{2}-1}+1\leq k\\
  \leq p^s-p^{s-\tau_{2}}+\\(\beta_{2}+1) p^{s-\tau_{2}-1}
\end{array}$ &$\min\{2(\beta_{1}+2)p^{\tau_{1}}, 4(\beta_{2}+2)p^{\tau_{2}}\}$\\ \hline
$i=p^s$&$j=p^s$&$\begin{array}{c}
  p^s-p^{s-\tau_{2}}+\\ \beta_{2} p^{s-\tau_{2}-1}+1\leq k\\
  \leq p^s-p^{s-\tau_{2}}+\\(\beta_{2}+1) p^{s-\tau_{2}-1}
\end{array}$ & $4(\beta_{2}+2)p^{\tau_{2}}$\\ \hline
$i=p^s$&$j=p^s$&$k=p^s$&$0$\\ \hline
\end{tabular}
\end{center}
\end{table}
\begin{table}\tiny
\begin{center}
\caption{The minimum Hamming distances of $\C=<{f_{1}(x)}^i(x-1)^j{f_{2}(x)}^k>$, where $0 \leq k\leq j \leq i \leq p^s$.}\label{Table:331}
\begin{tabular}{c|c|c|c}
    \hline  $i$ & $j$ & $k$ & $d_H(\C)$ \\ \hline
       $0$   &$0$ &$0$ &$1$\\\hline
        $0<i\leq p^{s-1}$ & $0\leq j\leq p^{s-1}$ &$k=0$  &$2$\\ \hline
$ p^{s-1}<i\leq p^{s}$ & $j=0$ &$k=0$  &$3$\\ \hline
$p^{s-1}<i\leq 2p^{s-1}$ & $0<j\leq p^{s}$ &$k=0$  &$3$\\ \hline
$2p^{s-1}<i\leq p^{s}$ & $0< j\leq p^{s}$ &$k=0$  &$4$\\ \hline
      $\begin{array}{c}
  p^s-p^{s-\tau_{0}}+\\ \beta_{0} p^{s-\tau_{0}-1}+1\leq i\\
  \leq p^s-p^{s-\tau_{0}}+\\(\beta_{0}+1) p^{s-\tau_{0}-1}
\end{array}$
 &  $\begin{array}{c}
  p^s-p^{s-\tau_{1}}+\\ \beta_{1} p^{s-\tau_{1}-1}+1\leq j\\
  \leq p^s-p^{s-\tau_{1}}+\\(\beta_{1}+1) p^{s-\tau_{1}-1}
\end{array}$ &
 $\begin{array}{c}
  p^s-p^{s-\tau_{2}}+\\ \beta_{2} p^{s-\tau_{2}-1}+1\leq k\\
  \leq p^s-p^{s-\tau_{2}}+\\(\beta_{2}+1) p^{s-\tau_{2}-1}
\end{array}$ & $\begin{array}{c}\min\{(\beta_{0}+2)p^{\tau_{0}}, 3(\beta_{1}+2)p^{\tau_{1}},\\ 4(\beta_{2}+2)p^{\tau_{2}}\}\end{array}$
  \\ \hline
$i=p^s$&$\begin{array}{c}
  p^s-p^{s-\tau_{1}}+\\ \beta_{1} p^{s-\tau_{1}-1}+1\leq j\\
  \leq p^s-p^{s-\tau_{1}}+\\(\beta_{1}+1) p^{s-\tau_{1}-1}
\end{array}$ & $\begin{array}{c}
  p^s-p^{s-\tau_{2}}+\\ \beta_{2} p^{s-\tau_{2}-1}+1\leq k\\
  \leq p^s-p^{s-\tau_{2}}+\\(\beta_{2}+1) p^{s-\tau_{2}-1}
\end{array}$ &$\min\{3(\beta_{1}+2)p^{\tau_{1}}, 4(\beta_{2}+2)p^{\tau_{2}}\}$\\ \hline
$i=p^s$&$j=p^s$&$\begin{array}{c}
  p^s-p^{s-\tau_{2}}+\\ \beta_{2} p^{s-\tau_{2}-1}+1\leq k\\
  \leq p^s-p^{s-\tau_{2}}+\\(\beta_{2}+1) p^{s-\tau_{2}-1}
\end{array}$ & $4(\beta_{2}+2)p^{\tau_{2}}$\\ \hline
$i=p^s$&$j=p^s$&$k=p^s$&$0$\\ \hline
\end{tabular}
\end{center}
\end{table}
\begin{table}\tiny
\begin{center}
\caption{The minimum Hamming distances of $\C=<(x-1)^i{f_{1}(x)}^j{f_{2}(x)}^k>$ or $<(x-1)^i{f_{2}(x)}^j{f_{1}(x)}^k>$, where $0 \leq i \leq j \leq k\leq p^s$.}\label{Table:34}
\begin{tabular}{c|c|c|c}
    \hline  $i$ & $j$ & $k$ & $d_H(\C)$ \\ \hline
       $0$   &$0$ &$0$ &$1$\\\hline
        $i=0$ & $0\leq j\leq p^{s-1}$ &$0<k\leq p^{s-1}$  &$2$\\ \hline
$i=0$ & $j=0$ &$p^{s-1}<k\leq p^{s}$  &$3$\\ \hline
 $i=0$ & $0< j\leq p^{s}$ &$p^{s-1}<k\leq 2p^{s-1}$  &$3$\\ \hline
$i=0$ & $0< j\leq p^{s}$ &$2p^{s-1}<k\leq3p^{s-1}$  &$4$\\ \hline
$i=0$ & $0<j\leq p^{s}$ &$3p^{s-1}<k\leq p^{s}$  &$5$\\ \hline
      $\begin{array}{c}
  p^s-p^{s-\tau_{2}}+\\ \beta_{2} p^{s-\tau_{2}-1}+1\leq i\\
  \leq p^s-p^{s-\tau_{2}}+\\(\beta_{2}+1) p^{s-\tau_{2}-1}
\end{array}$
 &  $\begin{array}{c}
  p^s-p^{s-\tau_{1}}+\\ \beta_{1} p^{s-\tau_{1}-1}+1\leq j\\
  \leq p^s-p^{s-\tau_{1}}+\\(\beta_{1}+1) p^{s-\tau_{1}-1}
\end{array}$ &
 $\begin{array}{c}
  p^s-p^{s-\tau_{0}}+\\ \beta_{0} p^{s-\tau_{0}-1}+1\leq k\\
  \leq p^s-p^{s-\tau_{0}}+\\(\beta_{0}+1) p^{s-\tau_{0}-1}
\end{array}$ & $\begin{array}{c}\min\{(\beta_{0}+2)p^{\tau_{0}}, 3(\beta_{1}+2)p^{\tau_{1}},\\ 5(\beta_{2}+2)p^{\tau_{2}}\}\end{array}$
  \\ \hline
$\begin{array}{c}
  p^s-p^{s-\tau_{2}}+\\ \beta_{2} p^{s-\tau_{2}-1}+1\leq i\\
  \leq p^s-p^{s-\tau_{2}}+\\(\beta_{2}+1) p^{s-\tau_{2}-1}
\end{array}$&$\begin{array}{c}
  p^s-p^{s-\tau_{1}}+\\ \beta_{1} p^{s-\tau_{1}-1}+1\leq j\\
  \leq p^s-p^{s-\tau_{1}}+\\(\beta_{1}+1) p^{s-\tau_{1}-1}
\end{array}$ & $k=p^s$ &$\min\{3(\beta_{1}+2)p^{\tau_{1}}, 5(\beta_{2}+2)p^{\tau_{2}}\}$\\ \hline
$\begin{array}{c}
  p^s-p^{s-\tau_{2}}+\\ \beta_{2} p^{s-\tau_{2}-1}+1\leq i\\
  \leq p^s-p^{s-\tau_{2}}+\\(\beta_{2}+1) p^{s-\tau_{2}-1}
\end{array}$&$j=p^s$&$k=p^s$ & $5(\beta_{2}+2)p^{\tau_{2}}$\\ \hline
$i=p^s$&$j=p^s$&$k=p^s$&$0$\\ \hline
\end{tabular}
\end{center}
\end{table}
\begin{table}\tiny
\begin{center}
\caption{The minimum Hamming distances of $\C=<{f_{1}(x)}^i(x-1)^j{f_{2}(x)}^k>$, where $0 \leq i \leq j \leq k\leq p^s$.}\label{Table:341}
\begin{tabular}{c|c|c|c}
    \hline  $i$ & $j$ & $k$ & $d_H(\C)$ \\ \hline
       $0$   &$0$ &$0$ &$1$\\\hline
        $i=0$ & $0\leq j\leq p^{s-1}$ &$0<k\leq p^{s-1}$  &$2$\\ \hline
$i=0$ & $j=0$ &$p^{s-1}<k\leq p^{s}$  &$3$\\ \hline
 $i=0$ & $0< j\leq p^{s}$ &$p^{s-1}<k\leq 2p^{s-1}$  &$3$\\ \hline
$i=0$ & $0< j\leq p^{s}$ &$2p^{s-1}<k\leq p^{s}$  &$4$\\ \hline
     $\begin{array}{c}
  p^s-p^{s-\tau_{2}}+\\ \beta_{2} p^{s-\tau_{2}-1}+1\leq i\\
  \leq p^s-p^{s-\tau_{2}}+\\(\beta_{2}+1) p^{s-\tau_{2}-1}
\end{array}$
 &  $\begin{array}{c}
  p^s-p^{s-\tau_{1}}+\\ \beta_{1} p^{s-\tau_{1}-1}+1\leq j\\
  \leq p^s-p^{s-\tau_{1}}+\\(\beta_{1}+1) p^{s-\tau_{1}-1}
\end{array}$ &
 $\begin{array}{c}
  p^s-p^{s-\tau_{0}}+\\ \beta_{0} p^{s-\tau_{0}-1}+1\leq k\\
  \leq p^s-p^{s-\tau_{0}}+\\(\beta_{0}+1) p^{s-\tau_{0}-1}
\end{array}$ & $\begin{array}{c}\min\{(\beta_{0}+2)p^{\tau_{0}}, 3(\beta_{1}+2)p^{\tau_{1}},\\ 4(\beta_{2}+2)p^{\tau_{2}}\}\end{array}$
  \\ \hline
$\begin{array}{c}
  p^s-p^{s-\tau_{2}}+\\ \beta_{2} p^{s-\tau_{2}-1}+1\leq i\\
  \leq p^s-p^{s-\tau_{2}}+\\(\beta_{2}+1) p^{s-\tau_{2}-1}
\end{array}$&$\begin{array}{c}
  p^s-p^{s-\tau_{1}}+\\ \beta_{1} p^{s-\tau_{1}-1}+1\leq j\\
  \leq p^s-p^{s-\tau_{1}}+\\(\beta_{1}+1) p^{s-\tau_{1}-1}
\end{array}$ & $k=p^s$ &$\min\{3(\beta_{1}+2)p^{\tau_{1}}, 4(\beta_{2}+2)p^{\tau_{2}}\}$\\ \hline
$\begin{array}{c}
  p^s-p^{s-\tau_{2}}+\\ \beta_{2} p^{s-\tau_{2}-1}+1\leq i\\
  \leq p^s-p^{s-\tau_{2}}+\\(\beta_{2}+1) p^{s-\tau_{2}-1}
\end{array}$&$j=p^s$&$k=p^s$ & $4(\beta_{2}+2)p^{\tau_{2}}$\\ \hline
$i=p^s$&$j=p^s$&$k=p^s$&$0$\\ \hline
\end{tabular}
\end{center}
\end{table}
 It is clear that all cyclic codes of length $5p^s$ have  form $\C=<(x-1)^{i_{0}}(x-\omega)^{i_{1}}(x-\omega^2)^{i_{2}}(x-\omega^3)^{i_{3}}(x-\omega^4)^{i_{4}}>$ with $0 \leq i_{0}, i_{1},i_{2}, i_{3}, i_{4}\leq p^s$, when $5\mid(q-1)$. If $i_{0}=i_{1}=i_{2}=i_{3}=i_{4}=0$, then $\C=\f_{q}[x]/(x^{5p^s}-1)$, $d_H(\C)=1$. If $i_{0}=i_{1}=i_{2}=i_{3}=i_{4}=p^s$, then $\C=\{\mathbf{0}\}$, $d_H(\C)=0$. These are the trivial cyclic codes of length $5p^s$ over $\f_q$. Next, we determine the minimum Hamming distances of all non-trivial cyclic codes of length $5p^s$ over $\f_q$. Without losing generality, we only consider the Hamming distances of code $\C$ when $i_{4}\leq i_{3}\leq i_{2}\leq i_{1}\leq i_{0}$.
\begin{lemma}\small
Let $i_{0}, i_{1}, i_{2}, i_{3}, i_{4}$ be integers such that $0<i_{0}\leq p^s$ and $0\leq i_{3}\leq i_{2}\leq i_{1}\leq p^s$, $i_{4}=0$, then
\begin{equation*}\label{222}
 d_{H}(\C)=\left\{
  \begin{array}{ll}
 2,&\mbox{ if } 0<i_{0}\leq p^{s-1}, 0\leq i_{3}\leq i_{2}\leq i_{1}\leq i_{0}\leq p^{s-1}, i_{4}=0,\\
&~\mbox{ or } p^{s-1}< i_{0}\leq p^s, i_{1}=i_{2}=i_{3}=i_{4}=0,\\
  3,&\mbox{ if } p^{s-1}< i_{0}\leq 2p^{s-1}, 0< i_{1}\leq p^s, 0\leq i_{3}\leq i_{2}\leq p^s, i_{4}=0,\\
&~\mbox{ or } 2p^{s-1}< i_{0}\leq p^s, 0< i_{1}\leq p^s, i_{2}=i_{3}=i_{4}=0,\\
4,&\mbox{ if } 2p^{s-1}< i_{0}\leq 3 p^{s-1}, 0< i_{2}\leq i_{1}\leq p^s, 0\leq i_{3}\leq p^s, i_{4}=0,\\
&~\mbox{ or }3p^{s-1}< i_{0}\leq p^s, 0< i_{2}\leq i_{1}\leq p^s, i_{3}=i_{4}=0,\\
&~\mbox{ or }3p^{s-1}< i_{0}\leq p^s, 0< i_{3}\leq i_{2}\leq i_{1}\leq p^{s-1}, i_{4}=0,\\
5,&\mbox{ if }
3p^{s-1}< i_{0}\leq p^{s}, p^{s-1 }< i_{1}\leq p^s, 0< i_{3}\leq i_{2}\leq p^s,  i_{4}=0.
\end{array}
\right.
\end{equation*}
\end{lemma}
\begin{proof}
\begin{itemize}
\item Case 1. $0< i_{0}\leq p^{s-1}$.

 According to Proposition \ref{p1} and $(\ref{25})$, we can get $ d_{H}(\C)=2$;

\item Case 2. $ p^{s-1} <i_{0}\leq p^s$.

If $ p^{s-1}<i_{0}\leq p^{s}$ and $i_{1}=i_{2}=i_{3}=0$, then according to Proposition \ref{p1} and $(\ref{25})$, we can get $ d_{H}(\C)=2$;

If $ p^{s-1}<i_{0}\leq 2p^{s-1}$, $0< i_{1}\leq p^s$ and $0\leq i_{3}\leq i_{2}\leq p^s$, then according to Proposition \ref{p1} and $(\ref{25})$, we can get $ d_{H}(\C)=3$;

If $ 2p^{s-1}<i_{0}\leq p^{s}$, $0< i_{1}\leq p^s$ and $i_{2}=i_{3}=0$, then according to Proposition \ref{p1} and  $(\ref{25})$, we can get $ d_{H}(\C)=3$;

If $2p^{s-1}< i_{0}\leq 3 p^{s-1}, 0< i_{2}\leq i_{1}\leq p^s$ and $0\leq i_{3}\leq p^s$, then according to Proposition \ref{p1} and $(\ref{25})$, we can get $ d_{H}(\C)=4$;

If $3p^{s-1}< i_{0}\leq p^s, 0< i_{2}\leq i_{1}\leq p^s$ and $i_{3}=0$, then according to Proposition \ref{p1} and $(\ref{25})$, we can get $ d_{H}(\C)=4$;

If $3p^{s-1}< i_{0}\leq p^s$ and $0< i_{3}\leq i_{2}\leq i_{1}\leq p^{s-1}$, then according to Proposition \ref{p1} and $(\ref{25})$, we can get $ d_{H}(\C)=4$;

If $3p^{s-1}< i_{0}\leq p^{s}, p^{s-1 }< i_{1}\leq p^s$ and $0< i_{3}\leq i_{2}\leq p^s$, then according to Proposition \ref{p1} and $(\ref{25})$, we can get $ d_{H}(\C)=5$;
\end{itemize}
According to the cases as above, the proof  is done.
\end{proof}
\begin{lemma}
Let $i_{0}, i_{1}, i_{2}, i_{3}, i_{4}$ be integers with $p^s-p^{s-\tau_{0}}+\beta_{0} p^{s-\tau_{0}-1}+1\leq i_{0}\leq p^s-p^{s-\tau_{0}}+(\beta_{0}+1) p^{s-\tau_{0}-1}$, $p^s-p^{s-\tau_{1}}+\beta_{1}p^{s-\tau_{1}-1}+1\leq i_{1}\leq p^s-p^{s-\tau_{1}}+(\beta_{1}+1) p^{s-\tau_{1}-1}$, $p^s-p^{s-\tau_{2}}+\beta_{2} p^{s-\tau_{2}-1}+1\leq i_{2}\leq p^s-p^{s-\tau_{2}}+(\beta_{2}+1) p^{s-\tau_{2}-1}$, $p^s-p^{s-\tau_{3}}+\beta_{3} p^{s-\tau_{3}-1}+1\leq i_{3}\leq p^s-p^{s-\tau_{3}}+(\beta_{3}+1) p^{s-\tau_{3}-1}$ and $p^s-p^{s-\tau_{4}}+\beta_{4} p^{s-\tau_{4}-1}+1\leq i_{4}\leq p^s-p^{s-\tau_{4}}+(\beta_{4}+1) p^{s-\tau_{4}-1}$,  where $0 \leq \beta_{0}, \beta_{1}, \beta_{2}, \beta_{3}, \beta_{4}\leq p -2$ and $0\leq\tau_{4}\leq\tau_{3}\leq\tau_{2} \leq\tau_{1} \leq\tau_{0}\leq s -1$.  Let $\C=<(x-1)^{i_{0}}(x-\omega)^{i_{1}}(x-\omega^2)^{i_{2}}(x-\omega^3)^{i_{3}}(x-\omega^4)^{i_{4}}>$, then $d_H(\C) =\min\{ (\beta_{0}+2)p^{\tau_{0}}, 2(\beta_{1}+2)p^{\tau_{1}}, 3(\beta_{2}+2)p^{\tau_{2}}, 4(\beta_{3}+2)p^{\tau_{3}}, 5(\beta_{4}+2)p^{\tau_{4}} \}$.
\end{lemma}
\begin{proof}
The proof is similar to that of Lemma \ref{28}. So we omit it.
\end{proof}
By the similar arguments as above, we can obtain the following lemmas immediately.
\begin{lemma}
Let $i_{0}, i_{1}, i_{2}, i_{3}, i_{4}$ be integers with $i_{0}=p^s$, $p^s-p^{s-\tau_{1}}+\beta_{1}p^{s-\tau_{1}-1}+1\leq i_{1}\leq p^s-p^{s-\tau_{1}}+(\beta_{1}+1) p^{s-\tau_{1}-1}$, $p^s-p^{s-\tau_{2}}+\beta_{2} p^{s-\tau_{2}-1}+1\leq i_{2}\leq p^s-p^{s-\tau_{2}}+(\beta_{2}+1) p^{s-\tau_{2}-1}$, $p^s-p^{s-\tau_{3}}+\beta_{3} p^{s-\tau_{3}-1}+1\leq i_{3}\leq p^s-p^{s-\tau_{3}}+(\beta_{3}+1) p^{s-\tau_{3}-1}$ and $p^s-p^{s-\tau_{4}}+\beta_{4} p^{s-\tau_{4}-1}+1\leq i_{4}\leq p^s-p^{s-\tau_{4}}+(\beta_{4}+1) p^{s-\tau_{4}-1}$,  where $0 \leq \beta_{1}, \beta_{2}, \beta_{3}, \beta_{4}\leq p -2$ and $0\leq\tau_{4}\leq\tau_{3}\leq\tau_{2} \leq\tau_{1} \leq s -1$.  Let $\C=<(x-1)^{p^s}(x-\omega)^{i_{1}}(x-\omega^2)^{i_{2}}(x-\omega^3)^{i_{3}}(x-\omega^4)^{i_{4}}>$, then $d_H(\C) =\min\{ (2(\beta_{1}+2)p^{\tau_{1}}, 3(\beta_{2}+2)p^{\tau_{2}}, 4(\beta_{3}+2)p^{\tau_{3}}, 5(\beta_{4}+2)p^{\tau_{4}} \}$.
\end{lemma}
\begin{lemma}
Let $i_{0}, i_{1}, i_{2}, i_{3}, i_{4}$ be integers with $i_{0}=i_{1}=p^s$, $p^s-p^{s-\tau_{2}}+\beta_{2} p^{s-\tau_{2}-1}+1\leq i_{2}\leq p^s-p^{s-\tau_{2}}+(\beta_{2}+1) p^{s-\tau_{2}-1}$, $p^s-p^{s-\tau_{3}}+\beta_{3} p^{s-\tau_{3}-1}+1\leq i_{3}\leq p^s-p^{s-\tau_{3}}+(\beta_{3}+1) p^{s-\tau_{3}-1}$ and $p^s-p^{s-\tau_{4}}+\beta_{4} p^{s-\tau_{4}-1}+1\leq i_{4}\leq p^s-p^{s-\tau_{4}}+(\beta_{4}+1) p^{s-\tau_{4}-1}$,  where $0 \leq \beta_{2}, \beta_{3}, \beta_{4}\leq p -2$ and $0\leq\tau_{4}\leq\tau_{3}\leq\tau_{2} \leq s -1$.  Let $\C=<(x-1)^{p^s}(x-\omega)^{p^s}(x-\omega^2)^{i_{2}}(x-\omega^3)^{i_{3}}(x-\omega^4)^{i_{4}}>$, then $d_H(\C) =\min\{ (3(\beta_{2}+2)p^{\tau_{2}}, 4(\beta_{3}+2)p^{\tau_{3}}, 5(\beta_{4}+2)p^{\tau_{4}} \}$.
\end{lemma}

\begin{lemma}
Let $i_{0}, i_{1}, i_{2}, i_{3}, i_{4}$ be integers with $i_{0}=i_{1}=i_{2}=p^s$, $p^s-p^{s-\tau_{3}}+\beta_{3} p^{s-\tau_{3}-1}+1\leq i_{3}\leq p^s-p^{s-\tau_{3}}+(\beta_{3}+1) p^{s-\tau_{3}-1}$ and $p^s-p^{s-\tau_{4}}+\beta_{4} p^{s-\tau_{4}-1}+1\leq i_{4}\leq p^s-p^{s-\tau_{4}}+(\beta_{4}+1) p^{s-\tau_{4}-1}$,  where $0 \leq \beta_{3}, \beta_{4}\leq p -2$ and $0\leq\tau_{4}\leq\tau_{3} \leq s -1$.  Let $\C=<(x-1)^{p^s}(x-\omega)^{p^s}(x-\omega^2)^{p^s}(x-\omega^3)^{i_{3}}(x-\omega^4)^{i_{4}}>$, then $d_H(\C) =\min\{ 4(\beta_{3}+2)p^{\tau_{3}}, 5(\beta_{4}+2)p^{\tau_{4}} \}$.
\end{lemma}
\begin{lemma}
Let $i_{0}, i_{1}, i_{2}, i_{3}, i_{4}$ be integers with $i_{0}=i_{1}=i_{2}=i_{3}=p^s$ and $p^s-p^{s-\tau_{4}}+\beta_{4} p^{s-\tau_{4}-1}+1\leq i_{4}\leq p^s-p^{s-\tau_{4}}+(\beta_{4}+1) p^{s-\tau_{4}-1}$,  where $0 \leq\beta_{4}\leq p -2$ and $0\leq\tau_{4}\leq s -1$.  Let $\C=<(x-1)^{p^s}(x-\omega)^{p^s}(x-\omega^2)^{p^s}(x-\omega^3)^{p^s}(x-\omega^4)^{i_{4}}>$, then $d_H(\C) = 5(\beta_{4}+2)p^{\tau_{4}}$.
\end{lemma}
According to the lemmas as above, we have the following theorem.
\begin{theorem}{\label{77}}
Let $p\geq 7$ be a prime satisfying $5\mid(q-1)$, $0 \leq \beta_{0}, \beta_{1}, \beta_{2}, \beta_{3}, \beta_{4}\leq p -2$ and $0\leq\tau_{4}\leq\tau_{3}\leq\tau_{2} \leq\tau_{1} \leq\tau_{0}\leq s -1$. Then cyclic codes of length $5p^s$ have  form $\C=<(x-1)^{i_{0}}(x-\omega)^{i_{1}}(x-\omega^2)^{i_{2}}(x-\omega^3)^{i_{3}}(x-\omega^4)^{i_{4}}>$. If $0\leq i_{4}\leq i_{3}\leq i_{2}\leq i_{1}\leq i_{0}\leq p^s$, then the minimum Hamming distances of $\C$ are  shown in Table \ref{Table:35}.\end{theorem}
\begin{table}\tiny
\begin{center}
\caption{The minimum Hamming distances of $\C=<(x-1)^{i_{0}}(x-\omega)^{i_{1}}(x-\omega^2)^{i_{2}}(x-\omega^3)^{i_{3}}(x-\omega^4)^{i_{4}}>$, where $0\leq i_{4}\leq i_{3}\leq i_{2}\leq i_{1}\leq i_{0}\leq p^s$.}\label{Table:35}
\begin{tabular}{c|c|c|c|c|c}
    \hline  $i_{0}$ & $i_{1}$ & $i_{2}$ &$i_{3}$ &$i_{4}$ &$d_H(\C)$ \\ \hline
       $0$   &$0$ &$0$&$0$   &$0$  &$1$\\ \hline
        $0< i_{0}\leq p^{s-1}$ &$0\leq i_{1}\leq p^{s-1}$ & $0\leq i_{2}\leq p^{s-1}$ &$0\leq i_{3}\leq p^{s-1}$ &$i_{4}=0$   &$2$\\ \hline
$p^{s-1}<i_{0}\leq p^{s}$ & $i_{1}=0$ &$i_{2}=0$&$i_{3}=0$&$i_{4}=0$&$2$\\ \hline
$p^{s-1}<i_{0}\leq 2p^{s-1}$ & $0< i_{1}\leq p^{s}$ & $0\leq i_{2}\leq p^{s}$ &$0\leq i_{3}\leq p^{s}$ &$i_{4}=0$   &$3$\\ \hline
 $2p^{s-1}<i_{0}\leq p^{s}$ & $0< i_{1}\leq p^s$ &$i_{2}=0$&$i_{3}=0$&$i_{4}=0$&$3$\\ \hline
  $2p^{s-1}< i_{0}\leq 3p^{s-1}$ &$0< i_{1}\leq p^{s}$ & $0<i_{2}\leq p^{s}$ &$0\leq i_{3}\leq p^{s}$ &$i_{4}=0$   &$4$\\ \hline
    $3p^{s-1}<i_{0}\leq p^{s}$ & $0< i_{1}\leq p^s$ &$0< i_{2}\leq p^s$&$i_{3}=0$&$i_{4}=0$&$4$\\ \hline
    $3p^{s-1}< i_{0}\leq p^s$ &  $0< i_{1} \leq p^{s-1}$ & $0< i_{2} \leq p^{s-1}$ & $0< i_{3} \leq p^{s-1}$ &$i_{4}=0$&$4$\\ \hline
 $3p^{s-1}<i_{0}\leq p^{s}$ & $p^{s-1}< i_{1}\leq p^s$ &$0< i_{2}\leq p^s$&$0< i_{3}\leq p^s$&$i_{4}=0$&$5$\\ \hline
$\begin{array}{c}
  p^s-p^{s-\tau_{0}}\\+ \beta_{0} p^{s-\tau_{0}-1}\\+1\leq i_{0}
  \leq p^s\\-p^{s-\tau_{0}}+\\(\beta_{0}+1) p^{s-\tau_{0}-1}
\end{array}$
 &  $\begin{array}{c}
  p^s-p^{s-\tau_{1}}\\+ \beta_{1} p^{s-\tau_{1}-1}\\+1\leq i_{1}
  \leq p^s\\-p^{s-\tau_{1}}+\\(\beta_{1}+1) p^{s-\tau_{1}-1}
\end{array}$ &
 $\begin{array}{c}
  p^s-p^{s-\tau_{2}}\\+\beta_{2} p^{s-\tau_{2}-1}\\+1\leq i_{2}
  \leq p^s\\-p^{s-\tau_{2}}+\\(\beta_{2}+1) p^{s-\tau_{2}-1}
\end{array}$ & $\begin{array}{c}
  p^s-p^{s-\tau_{3}}\\+\beta_{3} p^{s-\tau_{3}-1}\\+1\leq i_{3}
  \leq p^s\\-p^{s-\tau_{3}}+\\(\beta_{3}+1) p^{s-\tau_{3}-1}
\end{array}$ &$\begin{array}{c}
  p^s-p^{s-\tau_{4}}\\+\beta_{4} p^{s-\tau_{4}-1}\\+1\leq i_{4}
  \leq p^s\\-p^{s-\tau_{4}}+\\(\beta_{4}+1) p^{s-\tau_{4}-1}
\end{array}$ &$d_1$
  \\ \hline
$i_{0}=p^s$&$\begin{array}{c}
  p^s-p^{s-\tau_{1}}\\+ \beta_{1} p^{s-\tau_{1}-1}\\+1\leq i_{1}
  \leq p^s\\-p^{s-\tau_{1}}+\\(\beta_{1}+1) p^{s-\tau_{1}-1}
\end{array}$ &$\begin{array}{c}
  p^s-p^{s-\tau_{2}}\\+\beta_{2} p^{s-\tau_{2}-1}\\+1\leq i_{2}
  \leq p^s\\-p^{s-\tau_{2}}+\\(\beta_{2}+1) p^{s-\tau_{2}-1}
\end{array}$ &$\begin{array}{c}
  p^s-p^{s-\tau_{3}}\\+\beta_{3} p^{s-\tau_{3}-1}\\+1\leq i_{3}
  \leq p^s\\-p^{s-\tau_{3}}+\\(\beta_{3}+1) p^{s-\tau_{3}-1}
\end{array}$ &$\begin{array}{c}
  p^s-p^{s-\tau_{4}}\\+\beta_{4} p^{s-\tau_{4}-1}\\+1\leq i_{4}
  \leq p^s\\-p^{s-\tau_{4}}+\\(\beta_{4}+1) p^{s-\tau_{4}-1}
\end{array}$ &$d_2$
  \\ \hline
$i_{0}=p^s$&$i_{1}=p^s$&$\begin{array}{c}
  p^s-p^{s-\tau_{2}}\\+\beta_{2} p^{s-\tau_{2}-1}\\+1\leq i_{2}
  \leq p^s\\-p^{s-\tau_{2}}+\\(\beta_{2}+1) p^{s-\tau_{2}-1}
\end{array}$ &$\begin{array}{c}
  p^s-p^{s-\tau_{3}}\\+\beta_{3} p^{s-\tau_{3}-1}\\+1\leq i_{3}
  \leq p^s\\-p^{s-\tau_{3}}+\\(\beta_{3}+1) p^{s-\tau_{3}-1}
\end{array}$ &$\begin{array}{c}
  p^s-p^{s-\tau_{4}}\\+\beta_{4} p^{s-\tau_{4}-1}\\+1\leq i_{4}
  \leq p^s\\-p^{s-\tau_{4}}+\\(\beta_{4}+1) p^{s-\tau_{4}-1}
\end{array}$ &$d_3$
  \\ \hline
$i_{0}=p^s$&$i_{1}=p^s$&$i_{2}=p^s$&$\begin{array}{c}
  p^s-p^{s-\tau_{3}}\\+\beta_{3} p^{s-\tau_{3}-1}\\+1\leq i_{3}
  \leq p^s\\-p^{s-\tau_{3}}+\\(\beta_{3}+1) p^{s-\tau_{3}-1}
\end{array}$ &$\begin{array}{c}
  p^s-p^{s-\tau_{4}}\\+\beta_{4} p^{s-\tau_{4}-1}\\+1\leq i_{4}
  \leq p^s\\-p^{s-\tau_{4}}+\\(\beta_{4}+1) p^{s-\tau_{4}-1}
\end{array}$ &$d_4$
  \\ \hline
$i_{0}=p^s$&$i_{1}=p^s$&$i_{2}=p^s$&$i_{3}=p^s$&$\begin{array}{c}
  p^s-p^{s-\tau_{4}}\\+\beta_{4} p^{s-\tau_{4}-1}\\+1\leq i_{4}
  \leq p^s\\-p^{s-\tau_{4}}+\\(\beta_{4}+1) p^{s-\tau_{4}-1}
\end{array}$ &$d_5$
  \\ \hline
$i_{0}=p^s$&$i_{1}=p^s$&$i_{2}=p^s$&$i_{3}=p^s$&$i_{4}=p^s$&$0$\\ \hline
\end{tabular}
\end{center}
where $d_1=\min\{ (\beta_{0}+2)p^{\tau_{0}}, 2(\beta_{1}+2)p^{\tau_{1}}, 3(\beta_{2}+2)p^{\tau_{2}}, 4(\beta_{3}+2)p^{\tau_{3}}, 5(\beta_{4}+2)p^{\tau_{4}} \}$, $d_2=\min\{2(\beta_{1}+2)p^{\tau_{1}}, 3(\beta_{2}+2)p^{\tau_{2}}, 4(\beta_{3}+2)p^{\tau_{3}}, 5(\beta_{4}+2)p^{\tau_{4}} \}$, $d_3=\min\{3(\beta_{2}+2)p^{\tau_{2}}, 4(\beta_{3}+2)p^{\tau_{3}}, 5(\beta_{4}+2)p^{\tau_{4}} \}$, $d_4=\min\{ 4(\beta_{3}+2)p^{\tau_{3}}, 5(\beta_{4}+2)p^{\tau_{4}} \}$, $d_5=5(\beta_{4}+2)p^{\tau_{4}}$.
\end{table}

\section{MDS cyclic codes of length $5p^s$  over $\f_q$ }

In this section,  we obtain  all  MDS cyclic codes of length $5p^s$  over $\f_q$, where $s$ is an integer and $p\geq 7$ is a prime.

\begin{lemma}[Singleton Bound]\label{g2}
If  an $[n,\kappa,d]$ linear code over $\f_q$ exists, then $d\leq n-\kappa+1$.
\end{lemma}
A code for which equality holds in the Singleton Bound is called maximum distance
separable, abbreviated MDS.

\begin{theorem}
Suppose that $i_{4}\leq i_{3}\leq i_{2}\leq i_{1}\leq i_{0}$, then $\C=<(x-1)^{i_{0}}(x-\omega)^{i_{1}}(x-\omega^2)^{i_{2}}(x-\omega^3)^{i_{3}}(x-\omega^4)^{i_{4}}>$ over $\f_{q}$  is an MDS code if and only if one of the
following conditions holds:
\begin{itemize}
\item $i_{0}= i_{1}=i_{2}= i_{3}= i_{4}=0$. In this case, $d_{H}(\C)=1$.
\item $i_{0}= 1, i_{1}=i_{2}= i_{3}= i_{4}=0$. In this case, $d_{H}(\C)=2$.
\item $i_{0}= i_{1}=i_{2}= i_{3}=p^s,  i_{4}=p^{s}-1$. In this case, $d_{H}(\C)=5p^s$.
\end{itemize}
\end{theorem}
\begin{proof}
It is clear that  the dimension of code $\C$ is $5p^s-i_{0}- i_{1}-i_{2}- i_{3}- i_{4}$. By Singleton bound, $\C$ is an MDS code if and only if
$5p^s-i_{0}- i_{1}-i_{2}- i_{3}- i_{4}=5p^s-d_{H}(\C)+1$, i.e., $i_{0}+ i_{1}+i_{2}+i_{3}+ i_{4}=d_{H}(\C)-1$. The Hamming distances
of $\C$ for $0\leq i_{4}\leq i_{3}\leq i_{2}\leq i_{1}\leq i_{0}\leq p^s$ have been given in Theorem \ref{77}, then we can consider the conditions for the equations hold from the following eleven cases.

{\bf Case 1:} $i_{0}= i_{1}=i_{2}=i_{3}= i_{4}=0$. Then $d_{H}(\C)=1$, obviously, $0=d_{H}(\C)-1$.

{\bf Case 2:} $0<i_{0}\leq p^{s-1}, 0\leq i_{3}\leq i_{2}\leq i_{1}\leq i_{0}\leq p^{s-1}, i_{4}=0$, or
$ p^{s-1}< i_{0}\leq p^s, i_{1}=i_{2}=i_{3}=i_{4}=0$. Then $d_{H}(\C)=2$. Obviously, if $i_0=1$ and $i_{1}=i_{2}=i_{3}=i_{4}=0$, then $i_{0}+ i_{1}+i_{2}+i_{3}+ i_{4}=d_{H}(\C)-1$.

{\bf Case 3:} $p^{s-1}< i_{0}\leq 2p^{s-1}, 0< i_{1}\leq p^s, 0\leq i_{3}\leq i_{2}\leq p^s, i_{4}=0$, or
$2p^{s-1}< i_{0}\leq p^s, 0< i_{1}\leq p^s, i_{2}=i_{3}=i_{4}=0$. Then $d_{H}(\C)=3$, and $i_{0}+ i_{1}+i_{2}+i_{3}+ i_{4}>2$.

{\bf Case 4:} $p^{s-1}< i_{0}\leq 3 p^{s-1}, 0< i_{2}\leq i_{1}\leq p^s, 0\leq i_{3}\leq p^s, i_{4}=0$, or
$3p^{s-1}< i_{0}\leq p^s, 0< i_{2}\leq i_{1}\leq p^s, i_{3}=i_{4}=0$, or $3p^{s-1}< i_{0}\leq p^s, 0<i_{3}< i_{2}\leq i_{1}\leq p^{s-1}, i_{4}=0$.  Then $d_{H}(\C)=4$, and $i_{0}+ i_{1}+i_{2}+i_{3}+ i_{4}>3$.

{\bf Case 5:} $3p^{s-1}< i_{0}\leq p^{s}, p^{s-1}< i_{1}\leq p^s$,
 $0< i_{3}\leq i_{2}\leq p^s$ and $i_{4}=0$.  Then $d_{H}(\C)=5$, and $i_{0}+ i_{1}+i_{2}+i_{3}+ i_{4}>4$.

 {\bf Case 6:} $p^s-p^{s-\tau_{0}}+\beta_{0} p^{s-\tau_{0}-1}+1\leq i_{0}\leq p^s-p^{s-\tau_{0}}+(\beta_{0}+1) p^{s-\tau_{0}-1}$, $p^s-p^{s-\tau_{1}}+\beta_{1}p^{s-\tau_{1}-1}+1\leq i_{1}\leq p^s-p^{s-\tau_{1}}+(\beta_{1}+1) p^{s-\tau_{1}-1}$, $p^s-p^{s-\tau_{2}}+\beta_{2} p^{s-\tau_{2}-1}+1\leq i_{2}\leq p^s-p^{s-\tau_{2}}+(\beta_{2}+1) p^{s-\tau_{2}-1}$, $p^s-p^{s-\tau_{3}}+\beta_{3} p^{s-\tau_{3}-1}+1\leq i_{3}\leq p^s-p^{s-\tau_{3}}+(\beta_{3}+1) p^{s-\tau_{3}-1}$ and $p^s-p^{s-\tau_{4}}+\beta_{4} p^{s-\tau_{4}-1}+1\leq i_{4}\leq p^s-p^{s-\tau_{4}}+(\beta_{4}+1) p^{s-\tau_{4}-1}$,  where $0 \leq \beta_{0}, \beta_{1}, \beta_{2}, \beta_{3}, \beta_{4}\leq p -2$ and $0\leq\tau_{4}\leq\tau_{3}\leq\tau_{2} \leq\tau_{1} \leq\tau_{0}\leq s -1$. Then $d_H(\C) =\min\{ (\beta_{0}+2)p^{\tau_{0}}, 2(\beta_{1}+2)p^{\tau_{1}}, 3(\beta_{2}+2)p^{\tau_{2}}, 4(\beta_{3}+2)p^{\tau_{3}}, 5(\beta_{4}+2)p^{\tau_{4}} \}$, and

     \begin{eqnarray*}
    &&i_{0}+ i_{1}+i_{2}+i_{3}+ i_{4}\\
     &\geq&  5p^s-\sum^{4}\limits_{i=0}p^{s-\tau_{i}}+\sum^{4}\limits_{i=0}\beta_{i}p^{s-\tau_{i}-1}+5 \\
     &=& 5(p^{s-\tau_{0}}(p^{\tau_{0}}-1)+\beta_{0} p^{s-\tau_{0}-1}+1)\\ &~&(\mbox{equality when }
     \beta_{0}=\beta_{1}=\beta_{2}=\beta_{3}=\beta_{4}, \tau_{4}=\tau_{3}=\tau_{2}=\tau_{1}=\tau_{0})\\
      &\geq& 5p(p^{\tau_{0}}-1)+5\beta_{0}+5~ (\mbox{equality when } \tau_{0}=s-1) \\
     &\geq& 5(\beta_{0}+2)(p^{\tau_{0}}-1)+5\beta_{0}+5~( \mbox{equality  when } p=\beta_{0}+2) \\
      &=& (\beta_{0}+2)p^{\tau_{0}}+4(\beta_{0}+2)p^{\tau_{0}}-5 \\
      &>&(\beta_{0}+2)p^{\tau_{0}}-1\\
      &\geq& \min\{ (\beta_{0}+2)p^{\tau_{0}}, 2(\beta_{1}+2)p^{\tau_{1}}, 3(\beta_{2}+2)p^{\tau_{2}}, 4(\beta_{3}+2)p^{\tau_{3}}, 5(\beta_{4}+2)p^{\tau_{4}} \}-1\\
       &=& d_H(\C)-1.
 \end{eqnarray*}
Therefore, there is no MDS code.

 {\bf{Case 7:}} $i_{0}=p^s$, $p^s-p^{s-\tau_{1}}+\beta_{1}p^{s-\tau_{1}-1}+1\leq i_{1}\leq p^s-p^{s-\tau_{1}}+(\beta_{1}+1) p^{s-\tau_{1}-1}$, $p^s-p^{s-\tau_{2}}+\beta_{2} p^{s-\tau_{2}-1}+1\leq i_{2}\leq p^s-p^{s-\tau_{2}}+(\beta_{2}+1) p^{s-\tau_{2}-1}$, $p^s-p^{s-\tau_{3}}+\beta_{3} p^{s-\tau_{3}-1}+1\leq i_{3}\leq p^s-p^{s-\tau_{3}}+(\beta_{3}+1) p^{s-\tau_{3}-1}$ and $p^s-p^{s-\tau_{4}}+\beta_{4} p^{s-\tau_{4}-1}+1\leq i_{4}\leq p^s-p^{s-\tau_{4}}+(\beta_{4}+1) p^{s-\tau_{4}-1}$,  where $0 \leq \beta_{1}, \beta_{2}, \beta_{3}, \beta_{4}\leq p -2$ and $0\leq\tau_{4}\leq\tau_{3}\leq\tau_{2} \leq\tau_{1}\leq s -1$. Then $d_H(\C) =\min\{ 2(\beta_{1}+2)p^{\tau_{1}}, 3(\beta_{2}+2)p^{\tau_{2}}, 4(\beta_{3}+2)p^{\tau_{3}}, 5(\beta_{4}+2)p^{\tau_{4}} \}$, and

    \begin{eqnarray*}
     &&i_{0}+ i_{1}+i_{2}+i_{3}+ i_{4}\\ &\geq & 5p^s-\sum^{4}\limits_{i=1}p^{s-\tau_{i}}+\sum^{4}\limits_{i=1}\beta_{i}p^{s-\tau_{i}-1}+4 \\
    &=&p^s+ 4(p^{s-\tau_{1}}(p^{\tau_{1}}-1)+\beta_{1} p^{s-\tau_{1}-1}+1)\\
     &~&(\mbox{equality when }
      \beta_{1}=\beta_{2}=\beta_{3}=\beta_{4}, \tau_{4}=\tau_{3}=\tau_{2}=\tau_{1})\\
      &\geq& p^s+ 4p(p^{\tau_{1}}-1)+4\beta_{1}+4~(\mbox{equality when } \tau_{1}=s-1) \\
      &\geq& p^s+4(\beta_{1}+2)(p^{\tau_{1}}-1)+4\beta_{1}+4~(\mbox{equality when } p=\beta_{1}+2) \\
      &=& p^s+4(\beta_{1}+2)p^{\tau_{1}}-4\\
      &>&2(\beta_{1}+2)p^{\tau_{1}}-1\\
       &\geq& \min\{ 2(\beta_{1}+2)p^{\tau_{1}}, 3(\beta_{2}+2)p^{\tau_{2}}, 4(\beta_{3}+2)p^{\tau_{3}}, 5(\beta_{4}+2)p^{\tau_{4}} \}-1\\
       &=& d_H(\C)-1.
 \end{eqnarray*}
 Therefore, there is no MDS code.

  {\bf Case 8:} $i_{0}=i_{1}=p^s$,  $p^s-p^{s-\tau_{2}}+\beta_{2} p^{s-\tau_{2}-1}+1\leq i_{2}\leq p^s-p^{s-\tau_{2}}+(\beta_{2}+1) p^{s-\tau_{2}-1}$, $p^s-p^{s-\tau_{3}}+\beta_{3} p^{s-\tau_{3}-1}+1\leq i_{3}\leq p^s-p^{s-\tau_{3}}+(\beta_{3}+1) p^{s-\tau_{3}-1}$ and $p^s-p^{s-\tau_{4}}+\beta_{4} p^{s-\tau_{4}-1}+1\leq i_{4}\leq p^s-p^{s-\tau_{4}}+(\beta_{4}+1) p^{s-\tau_{4}-1}$,  where $0 \leq \beta_{2}, \beta_{3}, \beta_{4}\leq p -2$ and $0\leq\tau_{4}\leq\tau_{3}\leq\tau_{2}\leq s -1$. Then $d_H(\C) =\min\{ 3(\beta_{2}+2)p^{\tau_{2}}, 4(\beta_{3}+2)p^{\tau_{3}}, 5(\beta_{4}+2)p^{\tau_{4}} \}$, and

    \begin{eqnarray*}
     &&i_{0}+ i_{1}+i_{2}+i_{3}+ i_{4}\\ &\geq & 5p^s-\sum^{4}\limits_{i=2}p^{s-\tau_{i}}+\sum^{4}\limits_{i=2}\beta_{i}p^{s-\tau_{i}-1}+3 \\
     &=&2p^s+ 3(p^{s-\tau_{2}}(p^{\tau_{2}}-1)+\beta_{2} p^{s-\tau_{2}-1}+1)\\&~~~~~~~~~~&(\mbox{equality when }
      \beta_{2}=\beta_{3}=\beta_{4}, \tau_{4}=\tau_{3}=\tau_{2})\\
      &\geq& 2p^s+ 3p(p^{\tau_{2}}-1)+3\beta_{2}+3~(\mbox{equality when } \tau_{2}=s-1) \\
      &\geq& 2p^s+3(\beta_{2}+2)(p^{\tau_{2}}-1)+3\beta_{2}+3~(\mbox{equality when } p=\beta_{2}+2) \\
      &=& 2p^s+3(\beta_{2}+2)p^{\tau_{2}}-3\\
      &>&3(\beta_{2}+2)p^{\tau_{2}}-1\\
       &\geq& \min\{ 3(\beta_{2}+2)p^{\tau_{2}}, 4(\beta_{3}+2)p^{\tau_{3}}, 5(\beta_{4}+2)p^{\tau_{4}} \}-1\\
       &=& d_H(\C)-1.
 \end{eqnarray*}
 Therefore, there is no MDS code.

 {\bf{Case 9:}} $i_{0}=i_{1}=i_{2}=p^s$, $p^s-p^{s-\tau_{3}}+\beta_{3} p^{s-\tau_{3}-1}+1\leq i_{3}\leq p^s-p^{s-\tau_{3}}+(\beta_{3}+1) p^{s-\tau_{3}-1}$ and $p^s-p^{s-\tau_{4}}+\beta_{4} p^{s-\tau_{4}-1}+1\leq i_{4}\leq p^s-p^{s-\tau_{4}}+(\beta_{4}+1) p^{s-\tau_{4}-1}$,  where $0 \leq \beta_{3}, \beta_{4}\leq p -2$ and $0\leq\tau_{4}\leq\tau_{3}\leq s -1$. Then $d_H(\C) =\min\{  4(\beta_{3}+2)p^{\tau_{3}}, 5(\beta_{4}+2)p^{\tau_{4}} \}$, and

    \begin{eqnarray*}
     &&i_{0}+ i_{1}+i_{2}+i_{3}+ i_{4}\\ &\geq & 5p^s-\sum^{4}\limits_{i=3}p^{s-\tau_{i}}+\sum^{4}\limits_{i=3}\beta_{i}p^{s-\tau_{i}-1}+2 \\
     &=&3p^s+ 2(p^{s-\tau_{3}}(p^{\tau_{3}}-1)+\beta_{3} p^{s-\tau_{3}-1}+1)~(\mbox{equality when  }
     \beta_{3}=\beta_{4}, \tau_{4}=\tau_{3})\\
      &\geq& 3p^s+ 2p(p^{\tau_{3}}-1)+2\beta_{3}+2 ~(\mbox{equality when } \tau_{3}=s-1) \\
      &\geq& 3p^s+2(\beta_{3}+2)(p^{\tau_{3}}-1)+2\beta_{3}+2 ~(\mbox{equality when } p=\beta_{3}+2) \\
      &=& 3p^s+2(\beta_{3}+2)p^{\tau_{3}}-2\\
      &>&4(\beta_{3}+2)p^{\tau_{3}}-1\\
       &\geq& \min\{ 4(\beta_{3}+2)p^{\tau_{3}}, 5(\beta_{4}+2)p^{\tau_{4}} \}-1\\
       &=& d_H(\C)-1.
 \end{eqnarray*}
 Therefore, there is no MDS code.

 {\bf{Case 10:}} $i_{0}=i_{1}=i_{2}=i_{3}=p^s$,  $p^s-p^{s-\tau_{4}}+\beta_{4} p^{s-\tau_{4}-1}+1\leq i_{4}\leq p^s-p^{s-\tau_{4}}+(\beta_{4}+1) p^{s-\tau_{4}-1}$,  where $0 \leq \beta_{4}\leq p -2$ and $0\leq\tau_{4}\leq s -1$. Then $d_H(\C) =   5(\beta_{4}+2)p^{\tau_{4}}$, and

    \begin{eqnarray*}
    && i_{0}+ i_{1}+i_{2}+i_{3}+ i_{4} \\&\geq & 5p^s-p^{s-\tau_{4}}+\beta_{4}p^{s-\tau_{4}-1}+1 \\
     &=&4p^s+ (p^{s-\tau_{4}}(p^{\tau_{4}}-1)+\beta_{4} p^{s-\tau_{4}-1}+1)\\
      &\geq& 4p^{\tau_{4}+1}+ p(p^{\tau_{4}}-1)+\beta_{4}+1~ (\mbox{equality when } \tau_{4}=s-1) \\
      &\geq& 4(\beta_{4}+2)p^{\tau_{4}}+(\beta_{4}+2)(p^{\tau_{4}}-1)+\beta_{4}+1 ~(\mbox{equality when } p=\beta_{4}+2) \\
      &=& 5(\beta_{4}+2)p^{\tau_{4}}-1\\
      &=& d_H(\C)-1.
 \end{eqnarray*}
 Therefore, $i_{0}+ i_{1}+i_{2}+i_{3}+ i_{4}\geq 5(\beta_{4}+2)p^{\tau_{4}}-1$ with equality when $\beta_{4}=p-2, \tau_{4}=s-1$. (In this case $i_{0}=i_{1}=i_{2}=i_{3}=p^s$ and $i_{4}=p^s-1, d_H(\C)=5p^s$).

{\bf{Case 11:}} $i_{0}=i_{1}=i_{2}=i_{3}=i_{4}=p^s$. Then $d_H(\C)=0$, obviously, $5p^s>d_H(\C)-1$.

\end{proof}

By the similar arguments as above, we can obtain the following theorems immediately.
\begin{theorem}
Let $\C=<(x-1)^i{f_{1}(x)}^j{f_{2}(x)}^k>$, then $\C$ is an MDS code if and only if one of the
following conditions holds:
\begin{itemize}
\item $i = j = k=0$. In this case, $d_{H}(\C)=1$.
\item $i=1, j=k=0$. In this case, $d_{H}(\C)=2$.
\item  $i=p^s-1, j=k=p^s$. In this case, $d_{H}(\C)=5p^s$.
\end{itemize}
\end{theorem}

\begin{theorem}
Let $\C=<(x-1)^i{\phi_{5}(x)}^j>$, then  $\C$  is an MDS code if and only if one of the
following conditions holds:
\begin{itemize}
\item $i = j=0$. In this case, $d_{H}(\C)=1$.
\item $i=1,j=0$. In this case, $d_{H}(\C)=2$.
\end{itemize}
\end{theorem}

We now can list all repeated-root cyclic MDS codes of length $5p^s$ over $\f_q$ as follows.
\begin{theorem}
Let $\C$ be a repeated-root cyclic code of length $5p^s$ with generator polynomial $g(x)$. Then $\C$ is an MDS code if
and only if
\begin{itemize}
\item $deg(g(x))=0$. In this case, $d_{H}(\C)=1$.
\item $deg(g(x))=1$. In this case, $d_{H}(\C)=2$.
\item $deg(g(x))=5p^s-1$. In this case, $d_{H}(\C)=5p^s$.
\end{itemize}
\end{theorem}
\section{quantum synchronizable codes from $5p^s$-length cyclic codes}

 Luo and Ma \cite{L2018} showed that repeated-root cyclic codes
are indeed a good ingredient for constructing synchronizable
codes due to their advantageous dual-containing properties. In this section, we take quantum synchronizable codes from repeated-root
cyclic codes of length $5p^s$.

\begin{theorem}[\cite{L2018}]\label{lm}
Let $\C_1=<\prod\limits_{t\in T_l}{(M_t(x))}^{i_t}>$ and $\C_2=<\prod\limits_{t\in T_l}{(M_t(x))}^{j_t}>$ be cyclic codes of length $lp^s$, where $i_t , j_t$ are
integers in the range $0\leq j_t< i_t\leq p^s$, $i_t+i_{-t}\leq p^s$ and $j_t+j_{-t}\leq p^s$ for all $t\in T_l$. If there exists an integer $r\in T_l$ with $\gcd(r,l)=1$ satisfying either $i_r - j_r>p^{s-1}$ or $i_r - j_r>0$ and  $i_{r^{'}} - j_{r^{'}}>p^{s-1}$ for some $r^{'}\neq r\in T_l$, then for any pair of non-negative integers $a_l, a_r$ such that $a_l+a_r<lp^s$, there exists an $(a_l, a_r)-[[lp^s+a_l+a_r, lp^s-2\sum\limits_{t\in T_l}i_{t}\cdot|\C_{t,l}|]]_q$ quantum synchronizable code.
\end{theorem}
\begin{lemma}{\label{sc}}
\begin{itemize}
\item (1)
If $5\nmid(q^2-1)$, a cyclic code of length $5p^s$ is of the form $\C=<(x-1)^{i_1}{\phi_{5}(x)}^{i_2}>$, where $0\leq i_1,i_2\leq p^s$. The dual $\C^\bot$ of $\C$ is $\C^\bot=<(x-1)^{p^s-i_1}{\phi_{5}(x)}^{p^s-i_2}>$. Thus, $\C^\bot\subset \C$ if and only if $0\leq i_1,i_2\leq\frac{p^s-1}{2}$.

\item (2) If $5\mid(q+1)$ and $5\nmid(q-1)$, a cyclic code of length $5p^s$ is of the form $\C=<(x-1)^{i_1}{f_{1}(x)}^{i_2}{f_{2}(x)}^{i_3}>$, where $0\leq i_1,i_2, i_3\leq p^s$. The dual $\C^\bot$ of $\C$ is $\C^\bot=<(x-1)^{p^s-i_1}{f_{1}(x)}^{p^s-i_2}{f_{2}(x)}^{p^s-i_3}>$. Thus, $\C^\bot\subset \C$ if and only if $0\leq i_1\leq\frac{p^s-1}{2}$ and $0\leq i_2+i_3\leq p^s$.

\item (3) If $5\mid(q-1)$, a cyclic code of length $5p^s$ is of the form  $\C=<(x-1)^{i_{0}}(x-\omega)^{i_{1}}(x-\omega^2)^{i_{2}}(x-\omega^3)^{i_{3}}(x-\omega^4)^{i_{4}}>$, where $0 \leq i_{0}, i_{1},i_{2}, i_{3}, i_{4}\leq p^s$. The dual $\C^\bot$ of $\C$ is $\C^\bot=<(x-1)^{p^s-i_{0}}(x-\omega)^{p^s-i_{1}}(x-\omega^2)^{p^s-i_{2}}(x-\omega^3)^{p^s-i_{3}}(x-\omega^4)^{p^s-i_{4}}>$. Thus, $\C^\bot\subset \C$ if and only if $0\leq i_0\leq\frac{p^s-1}{2}$ and $0\leq i_1+i_2+i_3+i_4\leq p^s$.
\end{itemize}
\end{lemma}
Applying Lemma \ref{sc} to Theorem \ref{lm}, we can obtain quantum synchronizable codes.
\begin{theorem}
Let $\C_1=<(x-1)^{i_1}{\phi_{5}(x)}^{i_2}>$ and $\C_2=<(x-1)^{j_1}{\phi_{5}(x)}^{j_2}>$ be $[5p^s,5p^5-i_1-4i_2]_q$ and $[5p^s,5p^5-j_1-4j_2]_q$ cyclic codes, respectively, with $0\leq j_1<i_1\leq\frac{p^s-1}{2}$ and $0\leq j_2<i_2\leq\frac{p^s-1}{2}$. If $i_2-j_2>p^{s-1}$ or $i_2-j_2>0$ and $i_1-j_1>p^{s-1}$ holds,  then for any pair $a_l, a_r$ of non-negative integers such that $a_l+a_r<5p^s$, there exists an $(a_l, a_r)-[[5p^s+a_l+a_r, 5p^s-2(i_1+4i_2)]]_q$ quantum synchronizable code.
\end{theorem}
\begin{theorem}
Let $\C_1=<(x-1)^{i_1}{f_{1}(x)}^{i_2}{f_{2}(x)}^{i_3}>$ and $\C_2=<(x-1)^{j_1}{f_{1}(x)}^{j_2}{f_{2}(x)}^{j_3}>$ be $[5p^s,5p^s-i_1-2i_2-2i_3]_q$ and $[5p^s,5p^s-j_1-2j_2-2j_3]_q$ cyclic codes, respectively, with $0\leq j_1<i_1\leq\frac{p^s-1}{2}$ and $ j_2<i_2, j_3<i_3, 0\leq i_2+i_3\leq p^s, 0\leq j_2+j_3\leq p^s$. If there exists an integer $r\in T_5$  satisfying either $i_r - j_r>p^{s-1}$ or $i_r - j_r>0$ and  $i_{r^{'}} - j_{r^{'}}>p^{s-1}$ for some $r^{'}\neq r\in T_5$, then for any pair of non-negative integers $a_l, a_r$ such that $a_l+a_r<5p^s$, there exists an $(a_l, a_r)-[[5p^s+a_l+a_r, 5p^s-2(i_1+2i_2+2i_3)]]_q$ quantum synchronizable code.
\end{theorem}
\begin{theorem}
Let $\C=<(x-1)^{i_{0}}(x-\omega)^{i_{1}}(x-\omega^2)^{i_{2}}(x-\omega^3)^{i_{3}}(x-\omega^4)^{i_{4}}>$ and $\C=<(x-1)^{j_{0}}(x-\omega)^{j_{1}}(x-\omega^2)^{j_{2}}(x-\omega^3)^{j_{3}}(x-\omega^4)^{j_{4}}>$ be $[5p^s,5p^s-i_0-i_1-i_2-i_3-i_4]_q$ and $[5p^s,5p^s-j_0-j_1-j_2-j_3-j_4]_q$ cyclic codes, respectively, with $0\leq j_0<i_0\leq\frac{p^s-1}{2}$ and $j_1<i_1, j_2<i_2, j_3<i_3, j_4<i_4, 0\leq i_1+ i_2+i_3+i_4\leq p^s, 0\leq j_1+j_2+j_3+j_4\leq p^s$. If there exists an integer $r\in T_5$  satisfying either $i_r - j_r>p^{s-1}$ or $i_r - j_r>0$ and  $i_{r^{'}} - j_{r^{'}}>p^{s-1}$ for some $r^{'}\neq r\in T_5$, then for any pair of non-negative integers $a_l, a_r$ such that $a_l+a_r<5p^s$, there exists an $(a_l, a_r)-[[5p^s+a_l+a_r, 5p^s-2(i_0+i_1+i_2+i_3+i_4)]]_q$ quantum synchronizable code.
\end{theorem}

\section{Conclusion}
In this paper, we obtained the weight distributions of all cyclic code of length $5$ over $\f_q$ and the Hamming distances of all repeated-root  cyclic codes of length $5p^s$ over $\F_q$.
Furthermore,  we found
all MDS cyclic codes of length $5p^s$ and took quantum synchronizable codes from repeated-root
cyclic codes of length $5p^s$.

\end{document}